\documentclass[a4paper]{article}%
\usepackage{amsfonts}
\usepackage{amssymb}
\usepackage{amsmath}
\usepackage{graphicx}
\usepackage{hyperref}
\usepackage[margin=1in]{geometry}%
\providecommand{\U}[1]{\protect\rule{.1in}{.1in}}
\newtheorem{theorem}{Theorem}

\newtheorem{corollary}[theorem]{Corollary}

\newtheorem{definition}[theorem]{Definition}

\newtheorem{proposition}[theorem]{Proposition}
\newtheorem{remark}[theorem]{Remark}

\newenvironment{proof}[1][Proof]{\noindent\textbf{#1.} }{\ \rule{0.5em}{0.5em}}
\begin{document}

\title{Asymptotics of solutions for a basic case of\\fluid structure interaction}
\author{Christoph Boeckle\thanks{Supported by the Swiss National Science Foundation
(Grant No. 200021-124403).}~\thanks{Corresponding author.}\\{\small Theoretical Physics Department}\\{\small University of Geneva, Switzerland}\\{\small christoph.boeckle@unige.ch}
\and Peter Wittwer\thanks{Supported by the Swiss National Science Foundation (Grant
No. 200021-124403).}\\{\small Theoretical Physics Department}\\{\small University of Geneva, Switzerland}\\{\small peter.wittwer@unige.ch}}
\date{\today}
\maketitle

\begin{abstract}
We consider the Navier--Stokes equations in a half-plane with a drift term
parallel to the boundary and a small source term of compact support. We
provide detailed information on the behavior of the velocity and the vorticity
at infinity in terms of an asymptotic expansion at large distances from the
boundary. The expansion is universal in the sense that it only depends on the
source term through some multiplicative constants. This expansion is identical
to the one for the problem of an exterior flow around a small body moving at
constant velocity parallel to the boundary, and can be used as an artificial
boundary condition on the edges of truncated domains for numerical simulations.

\bigskip


\end{abstract}
\tableofcontents

\section{Introduction}

In what follows, we study the steady Navier--Stokes equations in the
half-plane $\Omega_{+}=\left\{  (x,y)\in\mathbb{R}^{2}\mid y>1\right\}  $ with
a drift term parallel to the boundary, a force of compact support, and zero
Dirichlet boundary conditions at the boundary of the half-plane and at
infinity.
\begin{align}
\partial_{x}\boldsymbol{u}+\boldsymbol{u}\cdot\mathbf{\nabla}\boldsymbol{u}%
+\mathbf{\nabla}p\mathbf{-}\Delta\boldsymbol{u} &  =\boldsymbol{F}%
~,\label{eq:nssteadyforce}\\
\mathbf{\nabla}\cdot\boldsymbol{u} &  =0~,\label{eq:incompressibility}%
\end{align}
where $\boldsymbol{F}$ is smooth and of compact support in $\Omega_{+}%
$,\textit{\ i.e.}, $\boldsymbol{F}\in C_{c}^{\infty}(\Omega_{+})$, subject to
the boundary conditions%
\begin{align}
\boldsymbol{u}(x,1) &  =0~,\hspace{1cm}x\in\mathbb{R}~,\label{eq:b0}\\
\lim\limits_{\mathbf{x\rightarrow\infty}}\boldsymbol{u}\mathbf{(x)} &
=0~.\label{eq:b1}%
\end{align}
For small forces, existence of a solution for this system together with basic
bounds on the decay at infinity was proved in
\cite{Hillairet.Wittwer-Existenceofstationary2009}, and uniqueness of
solutions was proved in \cite{Hillairet.Wittwer-Asymptoticdescriptionof2011}
in a very general context. In
\cite{Boeckle.Wittwer-Decayestimatessolutions2011} additional information on
the decay at infinity was obtained. See
\cite{Guo.etal-Leadingorderasymptotics2012}, where the velocity field has been
analyzed to leading order in a similar three dimensional case. For a general
introduction to the method used in this series of papers, see
\cite{Heuveline.Wittwer-ExteriorFlowsat2010}.

Note that the asymptotic behavior is identical to the one for the problem of
an exterior flow without force around a small body moving parallel to the wall
at constant velocity described in a frame comoving with the body (see
\cite{Hillairet.Wittwer-Asymptoticdescriptionof2011}). The explicit asymptotes
of the unique solution to (\ref{eq:nssteadyforce})--(\ref{eq:b1}) may thus in
particular be used as an artificial boundary condition for numerical
simulations of the\ aforementioned flow with a body, see
\cite{Boeckle.Wittwer-Artificialboundaryconditions2012}. Artificial boundary
conditions obtained this way have already been applied with success in the
numerical resolution of two and three-dimensional flows in the full space (see
\cite{Boenisch.etal-Adaptiveboundaryconditions2005},
\cite{Boenisch.etal-Secondorderadaptive2008},
\cite{Heuveline.Wittwer-AdaptiveBoundaryConditions2010} and
\cite{Latt.etal-Simulatingexteriordomain2006}).

In the remainder of this paper, when we invoke "the solution", we refer to the
solution constructed in \cite{Hillairet.Wittwer-Existenceofstationary2009},
\cite{Boeckle.Wittwer-Decayestimatessolutions2011} and
\cite{Hillairet.Wittwer-Asymptoticdescriptionof2011}.

Our main result is summarized in the following theorem.

\begin{theorem}
\label{thm:mainresult}Let $\boldsymbol{u}=(u,v)$ and $p$ be the solution to
equations (\ref{eq:nssteadyforce})--(\ref{eq:b1}) for $\boldsymbol{F}$ small
and let $\omega$ be the vorticity. Then, there exist constants $c_{1},c_{2}$
such that for $\varepsilon>0$,%
\begin{align}
\lim\limits_{y\rightarrow\infty}\sup\limits_{x\in\mathbb{R}}%
|y^{5/2-\varepsilon}(u(x,y)-u_{\mathrm{as}}(x,y))| &  =0~,\label{eq:asu}\\
\lim\limits_{y\rightarrow\infty}\sup\limits_{x\in\mathbb{R}}%
|y^{5/2-\varepsilon}(v(x,y)-v_{\mathrm{as}}(x,y))| &  =0~,\label{eq:asv}\\
\lim\limits_{y\rightarrow\infty}\sup\limits_{x\in\mathbb{R}}%
|y^{9/2-\varepsilon}(\omega(x,y)-\omega_{\mathrm{as}}(x,y))| &
=0~,\label{eq:asw}%
\end{align}
with%
\begin{align}
u_{\mathrm{as}}(x,y) &  =\frac{c_{1}}{y^{3/2}}\varphi_{1}(x/y)+\frac{c_{1}%
}{y^{2}}\varphi_{2,1}(x/y)+\frac{c_{2}}{y^{2}}\varphi_{2,2}(x/y)-\frac{c_{1}%
}{y^{2}}\eta_{W}(x/y^{2})-\frac{c_{1}}{y^{3}}\eta_{B}(x/y^{2}%
)~,\label{eq:asstructu}\\
v_{\mathrm{as}}(x,y) &  =\frac{c_{1}}{y^{3/2}}\psi_{1}(x/y)+\frac{c_{1}}%
{y^{2}}\psi_{2}(x/y)+\frac{c_{2}}{y^{2}}\psi_{2,2}(x/y)+\frac{c_{1}}{y^{3}%
}\omega_{W}(x/y^{2})+\frac{c_{1}}{y^{4}}\omega_{B}(x/y^{2}%
)~,\label{eq:asstructv}\\
\omega_{\mathrm{as}}(x,y) &  =\frac{c_{1}}{y^{3}}\omega_{W}(x/y^{2}%
)+\frac{c_{1}}{y^{4}}\omega_{B}(x/y^{2})~,\label{eq:asstructw}%
\end{align}
and functions $\varphi_{1}$, $\varphi_{2,1}$, $\varphi_{2,2}$, $\psi_{1}$,
$\psi_{2,1}$, $\psi_{2,2}$, $\eta_{W}$, $\eta_{B}$, $\omega_{W}$ and
$\omega_{B}$ as given in Appendix~\ref{sec:explicitasfun}.
\end{theorem}

\begin{remark}
This theorem is an immediate consequence of Theorem~\ref{thm:mainfourier} in
Section~\ref{sec:funeq}.
\end{remark}

\begin{itemize}

\item \textit{The functions }$\varphi_{1}$\textit{, }$\varphi_{2,1}$\textit{,
}$\varphi_{2,2}$\textit{, }$\psi_{1}$\textit{, }$\psi_{2,1}$\textit{, }%
$\psi_{2,2}$\textit{, }$\eta_{W}$\textit{, }$\eta_{B}$\textit{, }$\omega_{W}%
$\textit{ and }$\omega_{B}$\textit{ are universal, }i.e.\textit{, independent
of }$\boldsymbol{F}$\textit{.}

\item \textit{The power }$5/2$\textit{ in the limits (\ref{eq:asu}) and
(\ref{eq:asv}) is sharp, whereas the power }$9/2$\textit{ in (\ref{eq:asw})
can probably be improved by }$1/2$\textit{ at the price of additional
computations.}

\item \textit{Some terms in (\ref{eq:asstructu}) and (\ref{eq:asstructv}) are
unimportant in view of the limits (\ref{eq:asu}) and (\ref{eq:asv}), but they
are included such as to form a divergence-free velocity field in pairs of
successive terms of }$u_{\mathrm{as}}$\textit{ and }$v_{\mathrm{as}}$\textit{
and such as to have two orders in both of the two scalings }$x/y$\textit{ and
}$x/y^{2}$\textit{.}

\item \textit{The explicit forms of }$u_{\mathrm{as}}$\textit{ and
}$v_{\mathrm{as}}$\textit{ imply that}%
\begin{align*}
\lim\limits_{y\rightarrow\infty}y^{3/2}u(xy,y) &  =c_{1}\varphi_{1}(x)~,\\
\lim\limits_{y\rightarrow\infty}y^{3/2}v(xy,y) &  =c_{1}\psi_{1}(x)~,
\end{align*}
\textit{which shows that the bounds given in
\cite{Hillairet.Wittwer-Existenceofstationary2009} are sharp. Moreover, the
components of the velocity field associated to the functions }$\varphi_{i}%
$\textit{ and }$\psi_{i}$\textit{ are harmonic. The asymptotic expansion is
thus given by the superposition of a potential flow and a flow carrying the
vorticity, which is concentrated, to leading order, in a parabolic region
called the "wake", in the sense that}%
\[
\lim\limits_{y\rightarrow\infty}y^{3}\omega_{\mathrm{as}}(xy^{2}%
,y)=c_{1}\omega_{W}(x)~.
\]
\textit{In contrast to the case of an exterior problem in }$\mathbb{R}^{2}%
$\textit{ (see for example \cite{Boenisch.etal-Secondorderadaptive2008}), the
vorticity is however not exponentially small outside the wake, since we have
in particular, for all }$x\in\mathbb{R}$\textit{,}%
\[
\lim\limits_{y\rightarrow\infty}y^{4}\omega_{\mathrm{as}}(x,y)=c_{1}\omega
_{B}\left(  0\right)  \neq0~,
\]
\textit{which shows that a background of vorticity is created by the
interaction of the fluid with the boundary.}

\item \textit{This asymptotic expansion exhibits two scalings, whereas the
three dimensional analogue (see \cite{Guo.etal-Leadingorderasymptotics2012})
exhibits only one (the analogue to the }$x/y$\textit{ scaling). In addition,
the current expansion is sharp for all components of the velocity field and
takes into account an additional order, necessary to reveal the background of
vorticity outside the wake.}

\item \textit{The constants }$c_{1}$\textit{ and }$c_{2}$\textit{
are\ expressed in terms of the solution, in (\ref{eq:c1}) and (\ref{eq:c2})
respectively.}

\item \textit{These results confirm the conjecture concerning the vorticity of
the problem described in \cite{Hillairet.Wittwer-vorticityofOseen2008}. In the
present paper the asymptotic behavior is known modulo the constants }$c_{1}%
$\textit{ and }$c_{2}$\textit{, whereas the conjecture had three undetermined
constants in its representation.}
\end{itemize}

\bigskip The rest of this paper is organized as follows. In
Section~\ref{sec:functionSpaces} we recall the functional framework defined in
\cite{Hillairet.Wittwer-Existenceofstationary2009} in which the solutions were
constructed. In Section~\ref{sec:funeq} we also recall the map defined in
\cite{Hillairet.Wittwer-Existenceofstationary2009} which yielded the solution
in terms of its fixed point. We then present a new result which allows to
improve the bounds on the solution. In Section~\ref{sec:asymptoticTerms} we
first extract the leading order terms of the velocity and vorticity. Using
these terms, we then improve the bounds from Section~\ref{sec:funeq} and
extract the next order of the asymptotic expansion. The appendix contains an
explicit representation of the asymptotic terms, as well as various technical
propositions and details of computations used in the main sections.

\section{Functional framework\label{sec:functionSpaces}}

We first recall the functional framework of
\cite{Hillairet.Wittwer-Existenceofstationary2009}.

\begin{definition}
\label{def:fourier}Let $\hat{f}$ be a complex valued function on $\Omega_{+}$.
Then, we define the inverse Fourier transform $f=\mathcal{F}^{-1}[\hat{f}]$ by
the equation,
\[
f(x,y)=\mathcal{F}^{-1}[\hat{f}](x,y)=\frac{1}{2\pi}\int_{\mathbb{R}}%
e^{-ikx}\hat{f}(k,y)dk~,
\]
and $\hat{h}=\hat{f}\ast\hat{g}$ by%
\[
\hat{h}(k,y)=(\hat{f}\ast\hat{g})(k,y)=\frac{1}{2\pi}\int_{\mathbb{R}}\hat
{f}(k-k^{\prime},y)\hat{g}(k^{\prime},y)dk^{\prime}~,
\]
whenever the integrals make sense. We note that for functions $f,g$ which are
smooth and of compact support in $\Omega_{+}$ we have $f=\mathcal{F}^{-1}%
[\hat{f}]$, and that $fg=\mathcal{F}^{-1}[\hat{f}\ast\hat{g}]$, where%
\[
\hat{f}(k,y)=\mathcal{F}[f](k,y)=\int_{\mathbb{R}}e^{ikx}f(x,y)dx~,
\]
and similarly $\hat{g}=\mathcal{F}[g]$.
\end{definition}

Whereas in direct space we use the variables $(x,y)$, in Fourier space we use
the variables $(k,t)$, where $k$ is the Fourier-conjugated variable of $x$ and
$y\equiv t$ (this choice of notation was made to remain consistent with
\cite{Hillairet.Wittwer-Existenceofstationary2009}).

\begin{definition}
\label{def:mu}Let $\alpha$, $r\geq0$, $k\in\mathbb{R}$ and $t\geq1$, and let
\[
\mu_{\alpha,r}(k,t)=\frac{1}{1+\left(  |k|t^{r}\right)  ^{\alpha}}~.
\]
We set $\bar{\mu}_{\alpha}(k,t)=\mu_{\alpha,1}(k,t)$, $\tilde{\mu}_{\alpha
}(k,t)=\mu_{\alpha,2}(k,t)$.
\end{definition}

\begin{definition}
\label{def:BapqSpaces}We define, for fixed $\alpha\geq0$, and $p$, $q$
$\in\mathbb{R}$, $\mathcal{B}_{\alpha,p,q}$ to be the Banach space of
functions $\hat{f}\in C(\mathbb{R}\setminus\{0\}\times\lbrack1,\infty
),\mathbb{C})$, for which the norm%
\[
\Vert\hat{f};\mathcal{B}_{\alpha,p,q}\Vert=\sup_{t\geq1}\sup_{k\in
\mathbb{R}\setminus\{0\}}\frac{|\hat{f}(k,t)|}{\frac{1}{t^{p}}\bar{\mu
}_{\alpha}(k,t)+\frac{1}{t^{q}}\tilde{\mu}_{\alpha}(k,t)}%
\]
is finite. The notations $\mathcal{B}_{\alpha,p,\infty}$ and $\mathcal{B}%
_{\alpha,\infty,q}$ are used for spaces of functions for which the norms%
\[
\Vert\hat{f};\mathcal{B}_{\alpha,p,\infty}\Vert=\sup_{t\geq1}\sup
_{k\in\mathbb{R}\setminus\{0\}}\frac{|\hat{f}(k,t)|}{\frac{1}{t^{p}}\bar{\mu
}_{\alpha}(k,t)}%
\]
and%
\[
\Vert\hat{f};\mathcal{B}_{\alpha,\infty,q}\Vert=\sup_{t\geq1}\sup
_{k\in\mathbb{R}\setminus\{0\}}\frac{|\hat{f}(k,t)|}{\frac{1}{t^{q}}\tilde
{\mu}_{\alpha}(k,t)}%
\]
are finite, respectively.
\end{definition}

\begin{remark}
\textit{The following elementary properties of the spaces }$\mathcal{B}%
_{\alpha,p,q}$\textit{ will be routinely used without mention:}
\end{remark}

\begin{itemize}
\item for $\alpha\geq0$ and $p$, $q\in\mathbb{R}$, we have%
\[
\mathcal{B}_{\alpha,p,q}\subset\mathcal{B}_{\alpha,\min\{p,q\},\infty}~.
\]

\item \textit{if }$\alpha$\textit{, }$\alpha^{\prime}\geq0$\textit{, and }%
$p$\textit{, }$p^{\prime}$\textit{, }$q$\textit{, }$q^{\prime}$\textit{ }%
$\in\mathbb{R}$\textit{, then }%
\[
\mathcal{B}_{\alpha,p,q}\cap\mathcal{B}_{\alpha^{\prime},p^{\prime},q^{\prime
}}\subset\mathcal{B}_{\min\{\alpha^{\prime},\alpha,\},\min\{p^{\prime
},p\},\min\{q^{\prime},q\}}~.
\]

\end{itemize}

In the remainder of this paper, "$\mathrm{const.}$" stands for some constant
independent of $k$ and $t$ that may change from one occurrence to the next
without notice. If $\hat{f}\in\mathcal{B}_{\alpha,p,q}$ with $\alpha>1$, then
we have the bound%
\begin{align*}
\int_{\mathbb{R}}|\hat{f}(k,t)|dk  &  \leq\Vert\hat{f};\mathcal{B}%
_{\alpha,p,q}\Vert\int_{\mathbb{R}}\left(  \frac{1}{t^{p}}\bar{\mu}_{\alpha
}(k,t)+\frac{1}{t^{q}}\tilde{\mu}_{\alpha}(k,t)\right)  dk\\
&  \leq\mathrm{const.}~\Vert\hat{f};\mathcal{B}_{\alpha,p,q}\Vert\left(
\frac{1}{t^{p+1}}+\frac{1}{t^{q+2}}\right) \\
&  \leq\frac{\mathrm{const.}}{t^{\min\{p+1,q+2)}}\Vert\hat{f};\mathcal{B}%
_{\alpha,p,q}\Vert~,
\end{align*}
which by Definition~\ref{def:fourier}\ immediately gives%
\begin{equation}
\sup\limits_{x\in\mathbb{R}}|f(x,y)|\leq\frac{\mathrm{const.}}{y^{\min
\{p+1,q+2)}}\Vert\hat{f};\mathcal{B}_{\alpha,p,q}\Vert~.
\label{eq:supboundbyBapqbound}%
\end{equation}
The $\mathcal{B}_{\alpha,p,q}$ spaces thus encode the decay behavior in direct
space in the direction perpendicular to the wall, uniformly along lines
parallel to the wall. For convenience later on we also define%
\begin{align*}
\kappa &  =\sqrt{k^{2}-ik}~,\\
\tau &  =t-1~,\\
\sigma &  =s-1~,
\end{align*}
and%
\[
\Lambda_{-}=-\operatorname{Re}(\kappa)=-\frac{1}{2}\sqrt{2\sqrt{k^{2}+k^{4}%
}+2k^{2}}~.
\]
To further unburden the notations, we set%
\begin{align}
\mu_{0}  &  =\frac{1}{s^{7/2}}\bar{\mu}_{\alpha}(k,s)+\frac{1}{s^{3}}%
\tilde{\mu}_{\alpha}(k,s)~,\label{eq:defmu0}\\
\mu_{1}  &  =\frac{1}{s^{7/2}}\bar{\mu}_{\alpha}(k,s)+\frac{1}{s^{4}}%
\tilde{\mu}_{\alpha}(k,s)~. \label{eq:defmu1}%
\end{align}

\section{Functional equations\label{sec:funeq}}

We recall the definition of the maps given in
\cite{Hillairet.Wittwer-Existenceofstationary2009} which allowed to prove the
existence of a solution by the contraction mapping principle.\ We begin by
introducing the basic elements. The velocity field $(\hat{u},\hat{v})$ is
decomposed into%
\begin{align*}
\hat{u}  &  =-\hat{\eta}+\hat{\varphi}~,\\
\hat{v}  &  =\hat{\omega}+\hat{\psi}~,
\end{align*}
with $\hat{\omega}$ the vorticity. The nonlinear terms are represented by%
\begin{align}
\hat{Q}_{0}\left(  k,t\right)   &  =\hat{u}\ast\hat{\omega}+\hat{F}%
_{2}~,\label{eq:defQ0}\\
\hat{Q}_{1}\left(  k,t\right)   &  =\hat{v}\ast\hat{\omega}-\hat{F}_{1}~,
\label{eq:defQ1}%
\end{align}
where $\boldsymbol{\hat{F}}=(\hat{F}_{1},\hat{F}_{2})=\mathcal{F}%
[\boldsymbol{F}]$. The functions composing the velocity field are themselves
further decomposed as follows%
\begin{align}
\hat{\psi}  &  =\sum_{m=0,1}\sum_{n=1,2,3}\hat{\psi}_{n,m}~,\hspace{1cm}%
\hat{\varphi}=\sum_{m=0,1}\sum_{n=1,2,3}\hat{\varphi}_{n,m}~, \label{eq:comp1}%
\\
\hat{\omega}  &  =\sum_{m=0,1}\sum_{n=1,2,3}\hat{\omega}_{n,m}~,\hspace
{1cm}\hat{\eta}=\sum_{m=0,1}\sum_{n=1,2,3}\hat{\eta}_{n,m}~. \label{eq:comp2}%
\end{align}
For $\alpha>1$, we have the map%
\[%
\begin{array}
[c]{cccc}%
\mathcal{N}~\colon & \mathcal{V}_{\alpha} & \rightarrow & \mathcal{V}_{\alpha
}=\mathcal{B}_{\alpha,\frac{5}{2},1}\times\mathcal{B}_{\alpha,\frac{1}{2}%
,0}\times\mathcal{B}_{\alpha,\frac{1}{2},1}\\
& (\hat{\omega},\hat{u},\hat{v}) & \longmapsto & \mathcal{L}[\mathcal{C}%
[(\hat{\omega},\hat{u},\hat{v}),(\hat{\omega},\hat{u},\hat{v})]+(\hat{F}%
_{2},-\hat{F}_{1})]~,
\end{array}
\]
with%
\begin{equation}%
\begin{array}
[c]{cccc}%
\mathcal{C}~\colon & \mathcal{V}_{\alpha}\times\mathcal{V}_{\alpha} &
\rightarrow & \mathcal{W}_{\alpha}=\mathcal{B}_{\alpha,\frac{7}{2},\frac{5}%
{2}}\times\mathcal{B}_{\alpha,\frac{7}{2},\frac{5}{2}}\\
& ((\hat{\omega}_{1},\hat{u}_{1},\hat{v}_{1}),(\hat{\omega}_{2},\hat{u}%
_{2},\hat{v}_{2})) & \longmapsto & \left(  \hat{u}_{1}\ast\hat{\omega}%
_{2},\hat{v}_{1}\ast\hat{\omega}_{2}\right)  ~,
\end{array}
\label{eq:mapC}%
\end{equation}
a continuous bilinear map, and%
\begin{equation}%
\begin{array}
[c]{cccc}%
\mathcal{L}~\colon & \mathcal{W}_{\alpha} & \rightarrow & \mathcal{V}_{\alpha
}\\
& (\hat{Q}_{0},\hat{Q}_{1}) & \longmapsto & (\hat{\omega},\hat{u},\hat{v})~,
\end{array}
\label{eq:mapL}%
\end{equation}
a continuous linear map. The solution $(\hat{\omega},\hat{u},\hat{v})$ is
obtained, for $||(\hat{F}_{2},\hat{F}_{1});\mathcal{W}_{\alpha}||$
sufficiently small, as a fixed point of the map $\mathcal{N}$. Due to an
improved bound given in Appendix~\ref{sec:newConv}, tighter bounds on the
nonlinear terms $\hat{Q}_{0}$ and $\hat{Q}_{1}$ can be obtained.

\begin{proposition}
\label{prop:Zspace}Let $\alpha>1$. The bilinear map%
\[%
\begin{array}
[c]{cccc}%
\mathcal{C}~\colon & \mathcal{V}_{\alpha}\times\mathcal{V}_{\alpha} &
\rightarrow & \mathcal{Z}_{\alpha}=\mathcal{B}_{\alpha,\frac{7}{2},3}%
\times\mathcal{B}_{\alpha,\frac{7}{2},4}\\
& ((\hat{\omega}_{1},\hat{u}_{1},\hat{v}_{1}),(\hat{\omega}_{2},\hat{u}%
_{2},\hat{v}_{2})) & \longmapsto & \left(  \hat{u}_{1}\ast\hat{\omega}%
_{2},\hat{v}_{1}\ast\hat{\omega}_{2}\right)  ~,
\end{array}
\]
is continuous.
\end{proposition}

\begin{proof}
This is an immediate consequence of using Proposition~\ref{corr:optConvBapq}
of the present paper instead of Proposition~9 in
\cite{Hillairet.Wittwer-Existenceofstationary2009} in the proof of Lemma~4 in
\cite{Hillairet.Wittwer-Existenceofstationary2009}.
\end{proof}

Using Proposition~\ref{prop:Zspace}, most of the bounds on the functions in
(\ref{eq:comp1}) and (\ref{eq:comp2}) proved in
\cite{Hillairet.Wittwer-Existenceofstationary2009} are easily improved. In the
following proposition, we indicate in bold face all the indices which have
changed with respect to Propositions~12, 14, 16 and 18 of
\cite{Hillairet.Wittwer-Existenceofstationary2009}.

\begin{proposition}
\label{prop:BapqSpaces}Let $\alpha>1$, $\delta>0$. We then have
\begin{align*}%
\begin{tabular}
[c]{ll}%
$\hat{\psi}_{1,0}\in\mathcal{B}_{\alpha,\frac{\mathbf{3}}{\mathbf{2}%
}-\boldsymbol{\delta},\mathbf{2}}$ & $\hat{\psi}_{1,1}\in\mathcal{B}%
_{\alpha,\frac{1}{2},\mathbf{3}}$\\
$\hat{\psi}_{2,0}\in\mathcal{B}_{\alpha,\frac{5}{2},\mathbf{2}}$ & $\hat{\psi
}_{2,1}\in\mathcal{B}_{\alpha,\frac{5}{2},\mathbf{3}}$\\
$\hat{\psi}_{3,0}\in\mathcal{B}_{\alpha,\frac{5}{2},\mathbf{2}}$ & $\hat{\psi
}_{3,1}\in\mathcal{B}_{\alpha,\frac{5}{2},\mathbf{3}}$%
\end{tabular}
&
\begin{tabular}
[c]{ll}%
$\hat{\varphi}_{1,0}\in\mathcal{B}_{\alpha,\frac{\mathbf{3}}{\mathbf{2}%
}\mathbf{-}\boldsymbol{\delta},\mathbf{2}}$ & $\hat{\varphi}_{1,1}%
\in\mathcal{B}_{\alpha,\frac{1}{2},\mathbf{3}}$\\
$\hat{\varphi}_{2,0}\in\mathcal{B}_{\alpha,\frac{5}{2},\mathbf{2}}$ &
$\hat{\varphi}_{2,1}\in\mathcal{B}_{\alpha,\frac{5}{2},\mathbf{3}}$\\
$\hat{\varphi}_{3,0}\in\mathcal{B}_{\alpha,\frac{5}{2},\mathbf{2}}$ &
$\hat{\varphi}_{3,1}\in\mathcal{B}_{\alpha,\frac{5}{2},\mathbf{3}}$%
\end{tabular}
\\%
\begin{tabular}
[c]{ll}%
$\hat{\omega}_{1,0}\in\mathcal{B}_{\alpha,\frac{7}{2},\mathbf{3-}%
\boldsymbol{\delta}}$ & $\hat{\omega}_{1,1}\in\mathcal{B}_{\alpha,\frac{5}%
{2},1}$\\
$\hat{\omega}_{2,0}\in\mathcal{B}_{\alpha,\infty,3}$ & $\hat{\omega}_{2,1}%
\in\mathcal{B}_{\alpha,\frac{5}{2},\mathbf{3}}$\\
$\hat{\omega}_{3,0}\in\mathcal{B}_{\alpha,\frac{\mathbf{7}}{\mathbf{2}%
},\mathbf{3}}$ & $\hat{\omega}_{3,1}\in\mathcal{B}_{\alpha,\frac{5}%
{2},\mathbf{3}}$%
\end{tabular}
&
\begin{tabular}
[c]{ll}%
$\hat{\eta}_{1,0}\in\mathcal{B}_{\alpha,\frac{5}{2},\mathbf{2-}%
\boldsymbol{\delta}}$ & $\hat{\eta}_{1,1}\in\mathcal{B}_{\alpha,\frac{3}{2}%
,0}$\\
$\hat{\eta}_{2,0}\in\mathcal{B}_{\alpha,\infty,2}$ & $\hat{\eta}_{2,1}%
\in\mathcal{B}_{\alpha,\frac{5}{2},\mathbf{3}}$\\
$\hat{\eta}_{3,0}\in\mathcal{B}_{\alpha,\frac{5}{2},\mathbf{2}}$ & $\hat{\eta
}_{3,1}\in\mathcal{B}_{\alpha,\frac{3}{2},\mathbf{2}}$%
\end{tabular}
\end{align*}

\end{proposition}

\begin{remark}
\label{rem:propBapqSpaces}Given the decay behavior in direct space provided by
(\ref{eq:supboundbyBapqbound}), it is clear that the components with indices
$(1,1)$ play a dominant role in Theorem~\ref{thm:mainresult}. In fact,
functions in $\mathcal{B}_{\alpha,p,q}$ with $p\geq3/2$ and $q\geq1/2$ are
negligible in the sense of the limits given in (\ref{eq:asu}) and
(\ref{eq:asv}), although Theorem~\ref{thm:mainresult} includes some additional
terms to satisfy the divergence-free criterion and to have two orders of the
asymptotics in both scalings. In the same way, functions with indices $p\geq4$
and $q\geq3$ are negligible in the sense of the limit given in (\ref{eq:asw}).
One would then expect $\hat{\omega}_{2,1}$ and $\hat{\omega}_{3,1}$ to be
relevant, but new and better bounds are proved in
Section~\ref{sec:finalImprovements}, so that they will turn out to be
negligible, too.
\end{remark}

\begin{proof}
Using that $(\hat{Q}_{0},\hat{Q}_{1})\in\mathcal{Z}_{\alpha}$ and following
otherwise the proof of Lemma~5 in
\cite{Hillairet.Wittwer-Existenceofstationary2009}, this is straightforward
for all functions except $\hat{\omega}_{2,0}$, $\hat{\eta}_{2,0}$ and
$\hat{\omega}_{3,0}$. Note that $\delta\in(0,1)$ using (\ref{eq:boundlog}).

For $\hat{\omega}_{2,0}$, we recall that%
\[
\hat{\omega}_{2,0}(k,t)=\frac{1}{2}e^{-\kappa(t-1)}\int_{t}^{\infty}%
f_{2,0}(k,s-1)\hat{Q}_{0}(k,s)ds~,
\]
with%
\[
f_{2,0}(k,\sigma)=\left(  \frac{ik}{\kappa}-\frac{(|k|+\kappa)^{2}}{\kappa
}\right)  e^{-\kappa\sigma}+2(|k|+\kappa)e^{-|k|\sigma}~.
\]
We have the bound%
\[
\left\vert f_{2,0}(k,\sigma)\right\vert \leq\mathrm{const.}~(|k|^{1/2}%
+|k|)e^{-|k|\sigma}~,
\]
so that we therefore have for $\hat{\omega}_{2,0}$%
\begin{align}
\left\vert \hat{\omega}_{2,0}(k,t)\right\vert  &  \leq\mathrm{const.}%
~e^{\Lambda_{-}(t-1)}\int_{t}^{\infty}\left\vert f_{2,0}(k,\sigma)\right\vert
\mu_{0}(k,s)ds\nonumber\\
&  \leq\mathrm{const.}~e^{\Lambda_{-}(t-1)}e^{|k|(t-1)}\int_{t}^{\infty
}(|k|^{1/2}+|k|)e^{-|k|\sigma}\frac{1}{s^{7/2}}\bar{\mu}_{\alpha
}(k,s)ds\label{eq:omega20barbound}\\
&  +\mathrm{const.}~e^{\Lambda_{-}(t-1)}e^{|k|(t-1)}\int_{t}^{\infty
}(|k|^{1/2}+|k|)e^{-|k|\sigma}\frac{1}{s^{3}}\tilde{\mu}_{\alpha}(k,s)ds~.
\label{eq:omega20tildebound}%
\end{align}
The term in (\ref{eq:omega20barbound}) is estimated with
Proposition~\ref{prop:sgk3}%
\begin{equation}
e^{\Lambda_{-}(t-1)}e^{|k|(t-1)}\int_{t}^{\infty}(|k|^{1/2}+|k|)e^{-|k|\sigma
}\frac{1}{s^{7/2}}\bar{\mu}_{\alpha}(k,s)ds\leq\mathrm{const.}~e^{\Lambda
_{-}(t-1)}\frac{1}{t^{3}}\bar{\mu}_{\alpha}(k,t)~. \label{eq:omega20b1}%
\end{equation}
The term (\ref{eq:omega20tildebound}) requires us to distinguish the cases
$1\leq t\leq2$ and $t>2$. In the first case, we have, using
Proposition~\ref{prop:sgk3},
\begin{equation}
e^{\Lambda_{-}(t-1)}e^{|k|(t-1)}\int_{t}^{\infty}(|k|^{1/2}+|k|)e^{-|k|\sigma
}\frac{1}{s^{3}}\tilde{\mu}_{\alpha}(k,s)ds\leq\mathrm{const.}\frac{1}%
{t^{5/2}}\tilde{\mu}_{\alpha}(k,t)\leq\mathrm{const.}\frac{1}{t^{3}}\tilde
{\mu}_{\alpha}(k,t)~, \label{eq:omega20b2}%
\end{equation}
and in the second case we have, using (\ref{eq:kfortinmu}) to trade the factor
$|k|^{1/2}$ for a factor $s^{-1}$ and then applying
Proposition~\ref{prop:sgk3},%
\begin{align}
&  e^{\Lambda_{-}(t-1)}e^{|k|(t-1)}\int_{t}^{\infty}(|k|^{1/2}%
+|k|)e^{-|k|\sigma}\frac{1}{s^{3}}\tilde{\mu}_{\alpha}(k,s)ds\nonumber\\
&  \leq\mathrm{const.}~e^{\Lambda_{-}(t-1)}e^{|k|(t-1)}\int_{t}^{\infty
}e^{-|k|\sigma}\left(  \frac{1}{s^{4}}\tilde{\mu}_{\alpha-1/2}(k,s)+|k|\frac
{1}{s^{3}}\tilde{\mu}_{\alpha}(k,s)\right)  ds\nonumber\\
&  \leq\mathrm{const.}~e^{\Lambda_{-}(t-1)}\frac{1}{t^{3}}\tilde{\mu}%
_{\alpha-1/2}(k,t)~. \label{eq:omega20b3}%
\end{align}

Collecting (\ref{eq:omega20b1})--(\ref{eq:omega20b3}) and applying
(\ref{eq:mutomutilde}), we finally have
\[
\left\vert \hat{\omega}_{2,0}(k,t)\right\vert \leq\mathrm{const.}~\frac
{1}{t^{3}}\tilde{\mu}_{\alpha}(k,t)~.
\]
Indeed, for $t>2$ the index $\alpha$ is arbitrarily large due to the
exponential factor.

For the function $\hat{\eta}_{2,0}$, we have, from
\cite{Hillairet.Wittwer-Existenceofstationary2009} and using
Proposition~\ref{prop:Zspace}, that%
\[
|\hat{\eta}_{2,0}(k,t)|\leq\mathrm{const.}~e^{\Lambda_{-}(t-1)}e^{-|k|(t-1)}%
\left(  \frac{1}{t^{5/2}}\bar{\mu}_{\alpha}(k,t)+\frac{1}{t^{2}}\tilde{\mu
}_{\alpha}(k,t)\right)  ~.
\]
Using inequality (\ref{eq:mutomutilde}) shows that $\hat{\eta}_{2,0}%
\in\mathcal{B}_{\alpha,\infty,2}$.

For $\hat{\omega}_{3,0}$ we recall from
\cite{Hillairet.Wittwer-Existenceofstationary2009} that%
\[
|\hat{\omega}_{3,0}(k,t)|\leq\mathrm{const.}~\left\vert \frac{\kappa}%
{ik}(e^{\kappa(t-1)}-e^{-\kappa(t-1)})\right\vert \int_{t}^{\infty}%
|f_{3,0}(k,\sigma)|\mu_{0}(k,s)ds~,
\]
with%
\[
|f_{3,0}(k,\sigma)|\leq\mathrm{const.}~e^{\Lambda_{-}\sigma}\min
\{1,|\Lambda_{-}|^{2}\}\leq\mathrm{const.}~e^{\Lambda_{-}\sigma}|\Lambda
_{-}|~.
\]
Since $|\Lambda\_|\sim|k|^{1/2}$ for $|k|\leq1$ and $|\Lambda_{-}|\sim|k|$ for
$|k|>1$, we use for the first case $|f_{3,0}(k,\sigma)|\leq\mathrm{const.}%
~e^{\Lambda_{-}\sigma}|\Lambda_{-}|^{2}$ and for the second case
$|f_{3,0}(k,\sigma)|\leq\mathrm{const.}~e^{\Lambda_{-}\sigma}|\Lambda_{-}|$
and we have, for all $|k|$, using Proposition~\ref{prop:sgL3},%
\begin{align*}
|\hat{\omega}_{3,0}(k,t)|  &  \leq\mathrm{const.}~e^{|\Lambda_{-}|(t-1)}%
\int_{t}^{\infty}|\Lambda_{-}|\mu_{0}(k,s)ds\\
&  \leq\mathrm{const.}~\left(  \frac{1}{t^{7/2}}\bar{\mu}_{\alpha}%
(k,t)+\frac{1}{t^{3}}\tilde{\mu}_{\alpha}(k,t)\right)  ~.
\end{align*}

\end{proof}

\section{Asymptotic terms\label{sec:asymptoticTerms}}

\subsection{Strategy}

In this section we extract the leading asymptotic terms of the functions
$\hat{\psi}$, $\hat{\varphi}$, $\hat{\eta}$, $\hat{\omega}$ and $\partial
_{k}\hat{\omega}$. We then calculate an explicit representation of these
asymptotic terms in direct space which allows us to prove even tighter bounds
on the nonlinear terms $\hat{Q}_{0}$, $\hat{Q}_{1}$ as well as $\partial
_{k}\hat{Q}_{1}$, than the ones given in Proposition~\ref{prop:Zspace} and
\cite{Boeckle.Wittwer-Decayestimatessolutions2011}. The new bounds on $\hat
{Q}_{0}$ and $\hat{Q}_{1}$ are then used to further improve the bounds on
$\hat{\psi}$, $\hat{\varphi}$, $\hat{\eta}$, and $\hat{\omega}$, which,
together with the tighter bound on $\partial_{k}\hat{Q}_{1}$, allow us to
extract second-order terms in two steps. First, we extract the second order
terms of $\hat{\psi}$ and $\hat{\varphi}$, which allows us to improve the
bounds on the non-linear terms once again using their direct-space
representation. Then, we proceed to extract the second order terms of
$\hat{\eta}$ and $\hat{\omega}$.

The extraction procedure is as follows: we first identify the leading
components in view of Proposition~\ref{prop:BapqSpaces} and
Remark~\ref{rem:propBapqSpaces}. We then calculate for each of these
components the pointwise limit as $t\rightarrow\infty$ for one of two
scalings: $k\mapsto k/t$ if the slowest direct space decay in the sense of
(\ref{eq:supboundbyBapqbound}) is due to the index $p$, $k\mapsto k/t^{2}$ if
it is due to the index $q$. We finally prove that the difference between the
leading component and this pointwise limit is in a $\mathcal{B}_{\alpha
,p^{\prime},q}$ or $\mathcal{B}_{\alpha,p,q^{\prime}}$ which is smaller due to
an improvement in the index that determined the scaling choice, thus
identifying the pointwise limit as the leading asymptotic term. For the second
order asymptotic term, we proceed in the same way using any new bound obtained
in between to identify the components from which we have to extract it. As we
will see, this is actually the leading component minus the leading order
asymptotic term, for which we then calculate a new pointwise limit to obtain
the second order term

In this section, some bounds lead to a decrease of $\alpha$ by $-3$. Since the
solution exists for arbitrary $\alpha>3$, this does not pose a problem. We now
present our main technical result. To unburden the notation in the proofs and
results we set%
\begin{align*}
\alpha^{\prime}  &  =\alpha-1~,\\
\alpha^{\prime\prime}  &  =\alpha-2~.
\end{align*}

\begin{theorem}
[asymptotes in $\mathcal{B}_{\alpha,p,q}$ spaces]\label{thm:mainfourier}Let
$\hat{u}\in\mathcal{B}_{\alpha,\frac{1}{2},0}$, $\hat{v}\in\mathcal{B}%
_{\alpha,\frac{1}{2},1}$, $\hat{\omega}\in\mathcal{B}_{\alpha,\frac{5}{2},1}$
as constructed in \cite{Hillairet.Wittwer-Existenceofstationary2009}, with
$\hat{u}=-\hat{\eta}+\hat{\varphi}$, $\hat{v}=\hat{\omega}+\hat{\psi}$. We
then have, for $\alpha>4$, $\infty$ arbitrarily large and $\delta>0$,
\[%
\begin{tabular}
[c]{ll}%
$\hat{\psi}-\hat{\psi}_{\mathrm{as,1}}\in\mathcal{B}_{\alpha,1,\infty}$ &
$\hat{\psi}-\hat{\psi}_{\mathrm{as,1}}-\hat{\psi}_{\mathrm{as,2}}%
\in\mathcal{B}_{\alpha,\frac{3}{2}-\delta,\infty}$\\
$\hat{\varphi}-\hat{\varphi}_{\mathrm{as,1}}\in\mathcal{B}_{\alpha,1,\infty}$
& $\hat{\varphi}-\hat{\varphi}_{\mathrm{as,1}}-\hat{\varphi}_{\mathrm{as,2}%
}\in\mathcal{B}_{\alpha,\frac{3}{2}-\delta,\infty}$\\
$\hat{\omega}-\hat{\omega}_{\mathrm{as,1}}\in\mathcal{B}_{\alpha,\frac{7}%
{2}-\delta,2}$ & $\hat{\omega}-\hat{\omega}_{\mathrm{as,1}}-\hat{\omega
}_{\mathrm{as,2}}\in\mathcal{B}_{\alpha,\frac{7}{2}-\delta,3-\delta}$\\
$\hat{\eta}-\hat{\eta}_{\mathrm{as,1}}\in\mathcal{B}_{\alpha,\frac{5}%
{2}-\delta,1}$ & $\hat{\eta}-\hat{\eta}_{\mathrm{as,1}}-\hat{\eta
}_{\mathrm{as,2}}\in\mathcal{B}_{\alpha,\frac{5}{2}-\delta,2-\delta}$%
\end{tabular}
\ \ \ \ \ \ \ \
\]
where the functions with the subscripts \emph{"as"} and \emph{"as,2"} are
given as follows: for $\hat{\psi}$ by (\ref{eq:psias}) and (\ref{eq:psias2}),
for $\hat{\varphi}$ by (\ref{eq:phias}) and (\ref{eq:phias2}), for
$\hat{\omega}$ by (\ref{eq:omegaas}) and (\ref{eq:omegaas2}), and finally for
$\hat{\eta}$ by (\ref{eq:etaas}) and (\ref{eq:etaas2}).
\end{theorem}

In the remainder of this section we give a proof of this theorem.

\subsection{\label{sec:leadingorderpsiphi}Leading order in $\hat{\psi}$ and
$\hat{\varphi}$}

In view of Proposition~\ref{prop:BapqSpaces} and
Remark~\ref{rem:propBapqSpaces}, the leading order term of $\hat{\psi}$ \ and
$\hat{\varphi}$ are to be extracted from $\hat{\psi}_{1,1}$ and $\hat{\varphi
}_{1,1}$, respectively. We use that $\hat{\psi},\hat{\varphi}\in
\mathcal{B}_{\alpha,\frac{1}{2},\infty}\supset\mathcal{B}_{a,\frac{1}{2},2}$,
since for these functions we are not interested in the wake behavior. We have
(see \cite{Hillairet.Wittwer-Existenceofstationary2009}),%
\begin{align}
\hat{\psi}_{1,1}(k,t)  &  =\frac{1}{2}e^{-|k|(t-1)}\int_{1}^{t}h_{1,1}%
(k,s-1)\hat{Q}_{1}(k,s)ds~,\label{eq:psi11}\\
\hat{\varphi}_{1,1}(k,t)  &  =\frac{1}{2}e^{-|k|(t-1)}\int_{1}^{t}%
k_{1,1}(k,s-1)\hat{Q}_{1}(k,s)ds~, \label{eq:phi11}%
\end{align}
with%
\begin{align}
h_{1,1}(k,\sigma)  &  =-e^{|k|\sigma}+\frac{(|k|+\kappa)^{2}}{ik}%
e^{-|k|\sigma}-2\frac{\kappa(|k|+\kappa)}{ik}e^{-\kappa\sigma}~,
\label{eq:h11}\\
k_{1,1}(k,\sigma)  &  =-\frac{|k|}{ik}h_{1,1}(k,\sigma)~. \label{eq:k11}%
\end{align}
Formally, we get from (\ref{eq:psi11}) and (\ref{eq:phi11})
\begin{align*}
\lim_{t\rightarrow\infty}\sqrt{t}\hat{\psi}_{1,1}(k/t,t)  &  =-c_{1}\sqrt
{-ik}e^{-|k|}=:\hat{\psi}_{1,1}^{\mathrm{l}}(k)~,\\
\lim_{t\rightarrow\infty}\sqrt{t}\hat{\varphi}_{1,1}(k/t,t)  &  =c_{1}%
\frac{|k|}{ik}\sqrt{-ik}e^{-|k|}=:\hat{\varphi}_{1,1}^{\mathrm{l}}(k)~,
\end{align*}
with
\begin{equation}
c_{1}=\int_{1}^{\infty}\left(  s-1\right)  \hat{Q}_{1}\left(  0,s\right)  ds~.
\label{eq:c1}%
\end{equation}
This motivates the definition of the functions
\begin{align}
\hat{\psi}_{\mathrm{as,1}}(k,t)  &  =\frac{1}{\sqrt{t}}\hat{\psi}%
_{1,1}^{\mathrm{l}}(kt)=-c_{1}\sqrt{-ik}e^{-|k|t}~,\label{eq:psias}\\
\hat{\varphi}_{\mathrm{as,1}}(k,t)  &  =\frac{1}{\sqrt{t}}\hat{\varphi}%
_{1,1}^{\mathrm{l}}(kt)=c_{1}\frac{|k|}{ik}\sqrt{-ik}e^{-|k|t}~.
\label{eq:phias}%
\end{align}
Note that $\hat{\psi}_{\mathrm{as,1}},\hat{\varphi}_{\mathrm{as,1}}%
\in\mathcal{B}_{\alpha,\frac{1}{2},\infty}$. We now show that%
\begin{align}
\hat{\psi}_{1,1}-\hat{\psi}_{\mathrm{as,1}}  &  \in\mathcal{B}_{\alpha
^{\prime},1,\infty}~,\label{eq:psiRemainderSpace}\\
\hat{\varphi}_{1,1}-\hat{\varphi}_{\mathrm{as,1}}  &  \in\mathcal{B}%
_{\alpha^{\prime},1,\infty}~. \label{eq:phiRemainderSpace}%
\end{align}

\begin{proof}
We have
\[
\hat{\psi}_{1,1}=-\frac{|k|}{ik}\hat{\varphi}_{1,1}~,
\]
and thus all the bounds on $\hat{\psi}_{1,1}$ are directly transposable to
$\hat{\varphi}_{1,1}$, and we only present the proof for $\hat{\psi}_{1,1}$.
In order to prove (\ref{eq:psiRemainderSpace}) we analyze%
\begin{align*}
\hat{\psi}_{1,1}(k,t)-\hat{\psi}_{\mathrm{as,1}}(k,t)  &  =\frac{1}%
{2}e^{-|k|(t-1)}\int_{1}^{t}h_{1,1}(k,s-1)\hat{Q}_{1}(k,s)ds\\
&  +\frac{1}{2}e^{-|k|t}\int_{1}^{\infty}2\sqrt{-ik}(s-1)\hat{Q}_{1}(0,s)ds~.
\end{align*}
We rewrite this expression as a sum of terms which can easily be
bounded.\ Namely,%
\[
\hat{\psi}_{1,1}(k,t)-\hat{\psi}_{\mathrm{as,1}}(\hat{\psi},t)=\sum_{i=1}%
^{3}\hat{\psi}_{i}^{r,1}~,
\]
with%
\begin{align*}
\hat{\psi}_{1}^{r,1}  &  =\frac{1}{2}\left(  e^{-|k|(t-1)}-e^{-|k|t}\right)
\int_{1}^{t}h_{1,1}(k,s-1)\hat{Q}_{1}(k,s)ds~,\\
\hat{\psi}_{2}^{r,1}  &  =\frac{1}{2}e^{-|k|t}\int_{1}^{t}\left(
h_{1,1}(k,s-1)\hat{Q}_{1}(k,s)+2\sqrt{-ik}(s-1)\hat{Q}_{1}(0,s)\right)  ds~,\\
\hat{\psi}_{3}^{r,1}  &  =\frac{1}{2}e^{-|k|t}\int_{t}^{\infty}2\sqrt
{-ik}(s-1)\hat{Q}_{1}(0,s)ds~.
\end{align*}

To bound $\hat{\psi}_{1}^{r,1}$ we use that
\[
|h_{1,1}(k,\sigma)|\leq\mathrm{const.}~(1+|k|)e^{|k|\sigma}\min
\{1,(1+|k|^{1/2})|k|^{1/2}\sigma\}~,
\]
inequality (\ref{eq:kfortinmu}), Propositions~\ref{prop:sgk1} and
\ref{prop:sgk2}, so that we get%
\begin{align*}
|\hat{\psi}_{1}^{r,1}|  &  =\left\vert \frac{1}{2}\left(  e^{-|k|(t-1)}%
-e^{-|k|t}\right)  \int_{1}^{t}h_{1,1}(k,s-1)\hat{Q}_{1}(k,s)ds\right\vert \\
&  \leq\mathrm{const.}~e^{-|k|t}|k|e^{|k|}\int_{1}^{t}(1+|k|)e^{|k|\sigma}%
\min\{1,(1+|k|^{1/2})|k|^{1/2}\sigma\}\mu_{1}(k,s)ds\\
&  \leq\mathrm{const.}\left(  \frac{1}{t^{3/2}}\bar{\mu}_{\alpha-1}%
(k,t)+\frac{1}{t^{4}}\tilde{\mu}_{\alpha-1}(k,t)\right)  ~,
\end{align*}
which shows that $\hat{\psi}_{1}^{r,1}\in\mathcal{B}_{\alpha^{\prime},\frac
{3}{2},\infty}$.

To bound $\hat{\psi}_{2}^{r,1}$ we first note that by
(\ref{eq:meanValueTheoremQ1})%
\begin{align*}
&  h_{1,1}(k,\sigma)\hat{Q}_{1}(k,s)+2\sqrt{-ik}\sigma\hat{Q}_{1}(0,s)\\
&  =(h_{1,1}(k,\sigma)+2\sqrt{-ik}\sigma)\hat{Q}_{1}(k,s)-2\sqrt{-ik}%
k\sigma\partial_{k}\hat{Q}_{1}(\zeta,s)~,
\end{align*}
for some $\zeta\in\lbrack0,k]$. We analyze the expression%
\[
h_{1,1}(k,\sigma)+2\sqrt{-ik}\sigma=-e^{|k|\sigma}+\frac{(|k|+\kappa)^{2}}%
{ik}e^{-|k|\sigma}-2\frac{\kappa(|k|+\kappa)}{ik}e^{-\kappa\sigma}+2\sqrt
{-ik}\sigma
\]
in further detail, with $h_{1,1}$ given by (\ref{eq:h11}). A straightforward
bound is%
\begin{equation}
\left\vert h_{1,1}(k,\sigma)+2\sqrt{-ik}\sigma\right\vert \leq\mathrm{const.}%
~(1+|k|(\sigma+1))e^{|k|\sigma}~, \label{eq:psiBound1}%
\end{equation}
but since the leading terms cancel, we also have%
\begin{align*}
h_{1,1}(k,\sigma)+2\sqrt{-ik}\sigma &  =-\left(  e^{|k|\sigma}-1-|k|\sigma
\right)  -\left(  e^{-|k|\sigma}-1+|k|\sigma\right)  +2\left(  e^{-\kappa
\sigma}-1+\kappa\sigma\right) \\
&  +\frac{2|k|^{2}+2|k|\kappa}{ik}\left(  \left(  e^{-|k|\sigma}-1\right)
-\left(  e^{-\kappa\sigma}-1\right)  \right)  -2\kappa\sigma+2\sqrt{-ik}%
\sigma~,
\end{align*}
which we can bound, using (\ref{eq:kappa-sqrt(-ik)Bound}), by%
\begin{align}
\left\vert h_{1,1}(k,\sigma)+2\sqrt{-ik}\sigma\right\vert  &  \leq
\mathrm{const.}~|k|^{2}\sigma^{2}e^{|k|\sigma}+\mathrm{const.}~|k|^{2}%
\sigma^{2}+\mathrm{const.}~|\kappa|^{2}\sigma^{2}\nonumber\\
&  +\mathrm{const.}~(|k|^{1/2}+|k|)(|k|\sigma+|\kappa|\sigma)+\mathrm{const.}%
~|k|^{3/2}\sigma\nonumber\\
&  \leq\mathrm{const.}~\sigma(\sigma+1)(|k|+|k|^{2})e^{|k|\sigma}~.
\label{eq:psiBound2}%
\end{align}
We have used here, and shall routinely use again throughout this paper without
further explicit mention, that for all $z\in\mathbb{C}$ with
$\operatorname{Re}(z)\leq0$ and $N\in\mathbb{N}_{0}$,
\[
\left\vert \frac{e^{z}-\sum_{n=0}^{N}\frac{1}{n!}z^{n}}{z^{N+1}}\right\vert
\leq\mathrm{const.}~,
\]
and for all $z\in\mathbb{C}$ with $\operatorname{Re}(z)>0$
\[
\left\vert \frac{e^{z}-\sum_{n=0}^{N}\frac{1}{n!}z^{n}}{z^{N+1}}\right\vert
\leq\mathrm{const.}~e^{\operatorname{Re}(z)}~.
\]
Therefore, using (\ref{eq:psiBound1}) and (\ref{eq:psiBound2}), we get%
\begin{equation}
\left\vert h_{1,1}(k,\sigma)+2\sqrt{-ik}\sigma\right\vert \leq\mathrm{const.}%
~\min\{(1+|k|(\sigma+1)),(|k|+|k|^{2})\sigma(\sigma+1)\}e^{|k|\sigma}~.
\label{eq:psiBoundTot}%
\end{equation}
Collecting these bounds yields%
\begin{align*}
|\hat{\psi}_{2}^{r,1}|  &  =\left\vert \frac{1}{2}e^{-|k|t}\int_{1}^{t}\left(
h_{1,1}(k,s-1)\hat{Q}_{1}(k,s)+2\sqrt{-ik}(s-1)\hat{Q}_{1}(0,s)\right)
ds\right\vert \\
&  \leq\mathrm{const.}~e^{-|k|t}\int_{1}^{t}\left\vert h_{1,1}(k,\sigma
)+2\sqrt{-ik}\sigma\right\vert \mu_{1}(k,s)ds\\
&  +\mathrm{const.}~e^{-|k|t}|k|^{3/2}\int_{1}^{t}(s-1)\left\vert \partial
_{k}\hat{Q}_{1}(\zeta,s)\right\vert ds~.
\end{align*}
By (\ref{eq:dkomegaBound}) and (\ref{eq:kpexpabskt}) the second term is in
$\mathcal{B}_{\alpha,\frac{3}{2}-\delta,\infty}$. Using (\ref{eq:psiBoundTot})
and Propositions~\ref{prop:sgk1} and \ref{prop:sgk2} we also show that%
\begin{align*}
&  \left\vert e^{-|k|t}\int_{1}^{t}\left(  h_{1,1}(k,s-1)+2\sqrt
{-ik}(s-1)\right)  \hat{Q}_{1}(k,s)ds\right\vert \\
&  \leq\mathrm{const.}\left(  \frac{1}{t}\bar{\mu}_{\alpha}(k,t)+\frac
{1}{t^{5/2}}\bar{\mu}_{\alpha}(k,t)+\frac{1}{t^{3}}\tilde{\mu}_{\alpha
}(k,t)\right)  ~,
\end{align*}
such that, all in all, $\hat{\psi}_{2}^{r,1}\in\mathcal{B}_{\alpha,1,\infty}$.

Finally, using (\ref{eq:kpexpabskt}), we have%
\[
|\hat{\psi}_{3}^{r,1}|=\left\vert e^{-|k|t}\int_{t}^{\infty}\sqrt
{-ik}(s-1)\hat{Q}_{1}(0,s)ds\right\vert \leq\mathrm{const.}~e^{-|k|t}%
|k|^{1/2}\frac{1}{t^{3/2}}\in\mathcal{B}_{\alpha,2,\infty}~.
\]

Gathering the bounds on the $\hat{\psi}_{i}^{r,1}$ yields
(\ref{eq:psiRemainderSpace}), and by the opening remark of the proof also
(\ref{eq:phiRemainderSpace}).
\end{proof}

\subsection{Leading order in $\hat{\eta}$ and $\hat{\omega}$}

In view of Proposition~\ref{prop:BapqSpaces} and
Remark~\ref{rem:propBapqSpaces}, the leading order term of $\hat{\eta}$ and
$\hat{\omega}$ are to be extracted from $\hat{\eta}_{1,1}$ and $\hat{\omega
}_{1,1}$, respectively. We have (see
\cite{Hillairet.Wittwer-Existenceofstationary2009}),%
\begin{align}
\hat{\eta}_{1,1}(k,t) &  =\frac{1}{2}e^{-\kappa(t-1)}\int_{1}^{t}%
g_{1,1}(k,s-1)\hat{Q}_{1}(k,s)ds~,\label{eq:eta11}\\
\hat{\omega}_{1,1}(k,t) &  =\frac{1}{2}e^{-\kappa(t-1)}\int_{1}^{t}%
f_{1,1}(k,s-1)\hat{Q}_{1}(k,s)ds~,\label{eq:omega11}%
\end{align}
with%
\begin{align}
g_{1,1}\left(  k,\sigma\right)   &  =\frac{\kappa}{ik}\left(  e^{\kappa\sigma
}+\frac{(|k|+\kappa)^{2}}{ik}e^{-\kappa\sigma}-2\frac{|k|(|k|+\kappa)}%
{ik}e^{-|k|\sigma}\right)  ~,\label{eq:g11}\\
f_{1,1}\left(  k,\sigma\right)   &  =\frac{ik}{\kappa}g_{1,1}(k,\sigma
)~.\label{eq:f11}%
\end{align}
Formally, we get from (\ref{eq:eta11}) and (\ref{eq:omega11})%
\begin{align*}
\lim_{t\rightarrow\infty}\hat{\eta}_{1,1}(k/t^{2},t) &  =-c_{1}e^{-\sqrt{-ik}%
}=:\hat{\eta}_{1,1}^{\mathrm{l}}(k)~,\\
\lim_{t\rightarrow\infty}t\hat{\omega}_{1,1}(k/t^{2},t) &  =c_{1}\sqrt
{-ik}e^{-\sqrt{-ik}}=:\hat{\omega}_{1,1}^{\mathrm{l}}(k)~,
\end{align*}
with $c_{1}$ as defined in (\ref{eq:c1}). This motivates the definition of the
functions%
\begin{align}
\hat{\eta}_{\mathrm{as,1}}(k,t) &  =\hat{\eta}_{1,1}^{\mathrm{l}}%
(kt^{2})=-c_{1}e^{-\sqrt{-ik}t}~,\label{eq:etaas}\\
\hat{\omega}_{\mathrm{as,1}}(k,t) &  =\frac{1}{t}\hat{\omega}_{1,1}%
^{\mathrm{l}}(kt^{2})=c_{1}\sqrt{-ik}e^{-\sqrt{-ik}t}~.\label{eq:omegaas}%
\end{align}
Note that $\hat{\eta}_{\mathrm{as,1}}\in\mathcal{B}_{\alpha,\infty,0}$ and
$\hat{\omega}_{\mathrm{as,1}}\in\mathcal{B}_{\alpha,\infty,1}$. We now show
that%
\begin{align}
\hat{\eta}_{1,1}-\hat{\eta}_{\mathrm{as,1}} &  \in\mathcal{B}_{\alpha^{\prime
},\frac{3}{2},1}~,\label{eq:etaRemainderSpace}\\
\hat{\omega}_{1,1}-\hat{\omega}_{\mathrm{as,1}} &  \in\mathcal{B}%
_{\alpha^{\prime},\frac{5}{2},2}~.\label{eq:omegaRemainderSpace}%
\end{align}

\begin{proof}
We have%
\[
\hat{\omega}_{1,1}=\frac{ik}{\kappa}\hat{\eta}_{1,1}~,
\]
with, see Appendix~\ref{sec:technical},
\[
\mathrm{const.}\leq\left\vert \frac{ik}{\kappa}\right\vert \leq\mathrm{const.}%
~\min\{1,|\Lambda_{-}|\}~,
\]
which means that the bounds on $\hat{\omega}_{1,1}$ are the same as those for
$\hat{\eta}_{1,1}$ for $|k|>1$, but have an additional factor of $|\Lambda
_{-}|$ for $|k|\leq1$. This results in an increase of $1$ in both the indices
$p$ and $q$ for the components $\omega$ when compared to the ones for
$\hat{\eta}$. This means that $\hat{\omega}$ decays $1/t$ faster than
$\hat{\eta}$, and since%
\[
\lim_{t\rightarrow\infty}t\cdot\frac{ik}{t^{2}\kappa(k/t^{2})}=-\sqrt{-ik}~,
\]
the asymptote of $\hat{\omega}$ is naturally derived from the one of
$\hat{\eta}$. We therefore only present the details of the proof for
$\hat{\eta}$, since the proof for $\omega$ can easily be recovered by
inserting the appropriate factors in the proof for $\hat{\eta}$.

In order to prove (\ref{eq:etaRemainderSpace}) we set%
\[
\hat{\eta}_{1,1}(k,t)-\hat{\eta}_{\mathrm{as,1}}(k,t)=\sum_{i=1}^{4}\hat{\eta
}_{i}^{r,1}~,
\]
where%
\begin{align*}
\hat{\eta}_{1}^{r,1}  &  =\frac{1}{2}\left(  e^{-\kappa(t-1)}-e^{-\kappa
t}\right)  \int_{1}^{t}g_{1,1}(k,s-1)\hat{Q}_{1}(k,s)ds~,\\
\hat{\eta}_{2}^{r,1}  &  =\frac{1}{2}e^{-\kappa t}\int_{1}^{t}\left(
g_{1,1}(k,s-1)\hat{Q}_{1}(k,s)+2(s-1)\hat{Q}_{1}(0,s)\right)  ds~,\\
\hat{\eta}_{3}^{r,1}  &  =-\left(  e^{-\kappa t}-e^{-\sqrt{-ik}t}\right)
\int_{1}^{t}(s-1)\hat{Q}_{1}(0,s)ds~,\\
\hat{\eta}_{4}^{r,1}  &  =e^{-\sqrt{-ik}t}\int_{t}^{\infty}(s-1)\hat{Q}%
_{1}(0,s)ds~.
\end{align*}
We have%
\[
|g_{1,1}(k,\sigma)|\leq\left\{
\begin{array}
[c]{lc}%
\mathrm{const.}~\sigma e^{|\Lambda_{-}|\sigma} & \mathrm{for}~|k|\leq1\\
\mathrm{const.}~|\Lambda_{-}|e^{|\Lambda_{-}|\sigma} & \mathrm{for}~|k|>1
\end{array}
\right.  ,
\]
and we treat the two cases separately, using both times
Propositions~\ref{prop:sgL1} and \ref{prop:sgL2}. For $|k|\leq1$ we have%
\begin{align*}
|\hat{\eta}_{1}^{r,1}|  &  =\left\vert \frac{1}{2}\left(  e^{-\kappa
(t-1)}-e^{-\kappa t}\right)  \int_{1}^{t}g_{1,1}(k,s-1)\hat{Q}_{1}%
(k,s)ds\right\vert \\
&  \leq\mathrm{const.}~e^{\Lambda_{-}(t-1)}|\Lambda_{-}|\int_{1}^{t}\sigma
e^{|\Lambda_{-}|\sigma}\mu_{1}(k,s)ds\\
&  \leq\mathrm{const.}\left(  \frac{1}{t^{5/2}}\bar{\mu}_{\alpha}%
(k,t)+\frac{1}{t}\tilde{\mu}_{\alpha}(k,t)\right)  ~,
\end{align*}
and for $|k|>1$ we have, using (\ref{eq:kfortinmu}),%
\begin{align*}
|\hat{\eta}_{1}^{r,1}|  &  =\left\vert \frac{1}{2}\left(  e^{-\kappa
(t-1)}-e^{-\kappa t}\right)  \int_{1}^{t}g_{1,1}(k,s-1)\hat{Q}_{1}%
(k,s)ds\right\vert \\
&  \leq\mathrm{const.}~e^{\Lambda_{-}(t-1)}|\Lambda_{-}|\int_{1}^{t}%
|\Lambda_{-}|e^{|\Lambda_{-}|\sigma}\mu_{1}(k,s)ds\\
&  \leq\mathrm{const.}\left(  \frac{1}{t^{4}}\bar{\mu}_{\alpha-1}+\frac
{1}{t^{3}}\tilde{\mu}_{\alpha-1}\right)  ~,
\end{align*}
so that $\hat{\eta}_{1}^{r,1}\in\mathcal{B}_{\alpha^{\prime},\frac{5}{2},1}$.

To bound $\hat{\eta}_{2}^{r,1}$ we note that by (\ref{eq:meanValueTheoremQ1})%
\begin{align*}
&  g_{1,1}(k,\sigma)\hat{Q}_{1}(k,s)+2\sigma\hat{Q}_{1}(0,s)\\
&  =(g_{1,1}(k,\sigma)+2\sigma)\hat{Q}_{1}(k,s)-2\sigma k\partial_{k}\hat
{Q}_{1}(\zeta,s)~,
\end{align*}
for some $\zeta\in\lbrack0,k]$. We first analyze the expression%
\[
g_{1,1}(k,\sigma)+2\sigma=\frac{\kappa}{ik}\left(  e^{\kappa\sigma}%
+\frac{(|k|+\kappa)^{2}}{ik}e^{-\kappa\sigma}-2\frac{|k|(|k|+\kappa)}%
{ik}e^{-|k|\sigma}+2\frac{ik}{\kappa}\sigma\right)  ~.
\]
A straightforward bound is%
\begin{equation}
|g_{1,1}(k,\sigma)+2\sigma|\leq\left\{
\begin{array}
[c]{lc}%
\mathrm{const.}~\sigma e^{|\Lambda_{-}|\sigma} & \mathrm{for}~|k|\leq1\\
\mathrm{const.}~(\sigma+1+|\Lambda_{-}|)e^{|\Lambda_{-}|\sigma} &
\mathrm{for}~|k|>1
\end{array}
\right.  . \label{eq:etaBound1}%
\end{equation}
Since the leading terms cancel, we also have%
\begin{align}
g_{1,1}(k,\sigma)+2\sigma &  =\frac{\kappa}{ik}\left(  (e^{\kappa\sigma
}-1-\kappa\sigma)-(e^{-\kappa\sigma}-1+\kappa\sigma)\right) \nonumber\\
&  +\frac{\kappa}{ik}\left(  \frac{2|k|^{2}+2|k|\kappa}{ik}\left(
(e^{-\kappa\sigma}-1)-(e^{-|k|\sigma}-1)\right)  +2\kappa\sigma+2\frac
{ik}{\kappa}\sigma\right)  ~,\nonumber
\end{align}
which we can bound by%
\begin{align}
\left\vert g_{1,1}(k,\sigma)+2\sigma\right\vert  &  \leq\mathrm{const.}%
\left\vert \frac{\kappa}{ik}\right\vert \left(  |\Lambda_{-}|^{2}\sigma
^{2}e^{|\Lambda_{-}|\sigma}+|\Lambda_{-}|^{2}\sigma^{2}+\left\vert
\frac{2|k|^{2}+2|k|\kappa}{ik}\right\vert \left(  |\Lambda_{-}|\sigma
+|k|\sigma\right)  +2|\Lambda_{-}|^{2}\sigma\right) \nonumber\\
&  \leq\left\{
\begin{array}
[c]{lc}%
\mathrm{const.}~|\Lambda_{-}|\sigma(\sigma+1)e^{|\Lambda_{-}|\sigma} &
\mathrm{for}~|k|\leq1\\
\mathrm{const.}~|\Lambda_{-}|^{2}\sigma(\sigma+1)e^{|\Lambda_{-}|\sigma} &
\mathrm{for}~|k|>1
\end{array}
\right.  , \label{eq:etaBound2}%
\end{align}
using that%
\[
\left\vert \kappa+\frac{ik}{\kappa}\right\vert \leq\left\vert \frac{k^{2}%
-ik}{\kappa}+\frac{ik}{\kappa}\right\vert \leq\mathrm{const.}~|k|^{3/2}%
\leq\mathrm{const.}~|\Lambda_{-}|^{2}~.
\]
Therefore, using (\ref{eq:etaBound1}) and (\ref{eq:etaBound2}), we get%
\begin{equation}
\left\vert g_{1,1}(k,\sigma)+2\sigma\right\vert \leq\left\{
\begin{array}
[c]{lc}%
\mathrm{const.}~\sigma e^{|\Lambda_{-}|\sigma}\min\{1,|\Lambda_{-}%
|(\sigma+1)\} & \mathrm{for}~|k|\leq1\\
\mathrm{const.}~e^{|\Lambda_{-}|\sigma}\min\{(\sigma+1+|\Lambda_{-}%
|),|\Lambda_{-}|^{2}\sigma(\sigma+1)\} & \mathrm{for}~|k|>1
\end{array}
\right.  . \label{eq:etaBoundTot}%
\end{equation}
Collecting these bounds yields%
\begin{align*}
|\hat{\eta}_{2}^{r,1}|  &  =\left\vert \frac{1}{2}e^{-\kappa t}\int_{1}%
^{t}\left(  g_{1,1}(k,\sigma)\hat{Q}_{1}(k,s)+2\sigma\hat{Q}_{1}(0,s)\right)
ds\right\vert \\
&  \leq\mathrm{const.}~e^{\Lambda_{-}t}\int_{1}^{t}\left\vert g_{1,1}%
(k,\sigma)+2\sigma\right\vert \mu_{1}(k,s)ds+\mathrm{const.}~e^{\Lambda_{-}%
t}|k|\int_{1}^{t}\sigma\left\vert \partial_{k}\hat{Q}_{1}(\zeta,s)\right\vert
ds~.
\end{align*}
The second term of this inequality on $|\hat{\eta}_{2}^{r,1}|$ can be
integrated and bounded due to (\ref{eq:dkomegaBound}), and is in
$\mathcal{B}_{\alpha,\infty,\frac{3}{2}}$ by (\ref{eq:kqexpsquareikt}). For
the first term, using Propositions~\ref{prop:sgL1} and \ref{prop:sgL2} with
the bound (\ref{eq:etaBoundTot}) we have, for $|k|\leq1$,%
\begin{align*}
&  \left\vert \frac{1}{2}e^{-\kappa(t-1)}\int_{1}^{t}\left(  g_{1,1}%
(k,\sigma)+2\sigma\right)  \hat{Q}_{1}(k,s)ds\right\vert \\
&  \leq\mathrm{const.}\left(  \frac{1}{t}\tilde{\mu}_{\alpha}(k,t)+\frac
{1}{t^{3/2}}\bar{\mu}_{\alpha}(k,t)+\frac{1}{t^{2}}\tilde{\mu}_{\alpha
}(k,t)\right)  ~,
\end{align*}
and for $|k|>1$,%
\begin{align*}
&  \left\vert \frac{1}{2}e^{-\kappa(t-1)}\int_{1}^{t}\left(  g_{1,1}%
(k,\sigma)+2\sigma\right)  \hat{Q}_{1}(k,s)ds\right\vert \\
&  \leq\mathrm{const.}\left(  \frac{1}{t^{2}}\tilde{\mu}_{\alpha}%
(k,t)+\frac{1}{t^{3/2}}\bar{\mu}_{\alpha}(k,t)+\frac{1}{t^{2}}\tilde{\mu
}_{\alpha}(k,t)\right)  ~,
\end{align*}
which shows that $\hat{\eta}_{2}^{r,1}\in\mathcal{B}_{\alpha,\frac{3}{2},1}$.

We now bound $\hat{\eta}_{3}^{r,1}$, which, using (\ref{eq:kqexpsquareikt}),
yields%
\[
|\hat{\eta}_{3}^{r,1}|=\left\vert \left(  e^{-\kappa t}-e^{-\sqrt{-ik}%
t}\right)  \int_{1}^{t}(s-1)\hat{Q}_{1}(0,s)ds\right\vert \leq\mathrm{const.}%
\left\vert e^{-\sqrt{-ik}t}\right\vert |k|^{3/2}t\in\mathcal{B}_{\alpha
,\infty,2}~.
\]

Finally, using (\ref{eq:kqexpsquareikt}), we have%
\[
|\hat{\eta}_{4}^{r,1}|=\left\vert e^{-\sqrt{-ik}t}\int_{t}^{\infty}%
(s-1)\hat{Q}_{1}(0,s)ds\right\vert \leq\mathrm{const.}\left\vert
e^{-\sqrt{-ik}t}\right\vert \frac{1}{t^{3/2}}\in\mathcal{B}_{\alpha
,\infty,\frac{3}{2}}~.
\]

Gathering the bounds on the $\hat{\eta}_{i}^{r,1}$ yields
(\ref{eq:etaRemainderSpace}), and by the opening remark of the proof also
(\ref{eq:omegaRemainderSpace}).
\end{proof}

\subsection{Leading order in $\partial_{k}\hat{\omega}$}

For technical reasons that will become clear in the procedure of extracting
second order asymptotic terms, it is necessary to give tighter bounds on
$\partial_{k}\hat{Q}_{1}=\hat{v}\ast\partial_{k}\hat{\omega}+\partial_{k}%
\hat{F}_{1}$ and $\partial_{k}\hat{\omega}$ (we recall that $\hat{\omega}$ is
continuous on $\mathbb{R}$ and $C^{1}$ on $\mathbb{R}\backslash\{0\}$, and
that the derivative on $\mathbb{R}$ is to be understood in the sense of
distributions). From \cite{Boeckle.Wittwer-Decayestimatessolutions2011} we
have%
\begin{align}
\kappa\partial_{k}\hat{\omega} &  \in\mathcal{B}_{\alpha^{\prime},\frac{3}%
{2},0}~,\label{eq:dkomegaKappaBound}\\
\kappa\partial_{k}\hat{\omega}-\kappa\partial_{k}\hat{\omega}_{1,1,1}%
-\kappa\partial_{k}\hat{\omega}_{2,1,1} &  \in\mathcal{B}_{\alpha^{\prime
},\frac{3}{2},1}~,\nonumber
\end{align}
with%
\begin{align}
\partial_{k}\hat{\omega}_{1,1,1}(k,t) &  =\frac{1}{2}\left(  \partial
_{k}e^{-\kappa\tau}\right)  \int_{1}^{t}f_{1,1}(k,s-1)\hat{Q}_{1}%
(k,s)ds~,\label{eq:dkomega111}\\
\partial_{k}\hat{\omega}_{2,1,1}(k,t) &  =\frac{1}{2}e^{-\kappa\tau}\int
_{1}^{t}\left(  \partial_{k}f_{1,1}(k,s-1)\right)  \hat{Q}_{1}%
(k,s)ds~,\label{eq:dkomega211}%
\end{align}
with $f_{1,1}$ given by (\ref{eq:f11}), with%
\[
\partial_{k}\kappa=\frac{2k-i}{2\kappa}~,
\]
and%
\begin{align}
\partial_{k}f_{1,1}(k,\sigma) &  =i\frac{\left(  |k|+\kappa\right)  ^{2}%
}{\kappa|k|}(e^{-|k|\sigma}-e^{-\kappa\sigma})\nonumber\\
&  +\frac{k^{2}+\kappa^{2}}{2\kappa k}(e^{\kappa\sigma}+e^{-\kappa\sigma
})\sigma\nonumber\\
&  +2i\frac{k^{2}+|k|\kappa}{k^{2}}\left(  \frac{k^{2}+\kappa^{2}}{2\kappa
}e^{-\kappa\sigma}-|k|e^{-|k|\sigma}\right)  \sigma~.\label{eq:dkf11}%
\end{align}
We have, from (\ref{eq:omegaas}),
\begin{equation}
\partial_{k}\hat{\omega}_{\mathrm{as}}(k)=i\frac{c_{1}}{2}\left(  1-\frac
{1}{\sqrt{-ik}t}\right)  te^{-\sqrt{-ik}t}~,\label{eq:dkomegaas}%
\end{equation}
with $c_{1}$ as defined by (\ref{eq:c1}). Note that $\partial_{k}\hat{\omega
}_{\mathrm{as}}\in\mathcal{B}_{\alpha,\infty,0}$. We now show that%
\begin{equation}
\kappa\partial_{k}\hat{\omega}_{1,1,1}+\kappa\partial_{k}\hat{\omega}%
_{2,1,1}-\kappa\partial_{k}\hat{\omega}_{\mathrm{as}}\in\mathcal{B}%
_{\alpha^{\prime\prime},\frac{3}{2},1}~.\label{eq:dkomegaRemainderSpace}%
\end{equation}

\begin{remark}
Note that%
\[
\mathcal{F}^{-1}[-i\partial_{k}\hat{\omega}_{\mathrm{as}}(k,y)]=x\omega
_{W}(x,y)=x\mathcal{F}^{-1}[\hat{\omega}_{\mathrm{as,1}}(k,y)]~.
\]

\end{remark}

\begin{proof}
In order to prove (\ref{eq:dkomegaRemainderSpace}) we note that%
\begin{align}
&  \kappa\partial_{k}\hat{\omega}_{1,1,1}(k,t)+\kappa\partial_{k}\hat{\omega
}_{2,1,1}(k,t)-\kappa\partial_{k}\hat{\omega}_{\mathrm{as}}(k,t)\nonumber\\
&  =-\frac{2k-i}{4}\tau e^{-\kappa\tau}\int_{1}^{t}f_{1,1}(k,s-1)\hat{Q}%
_{1}(k,s)ds\nonumber\\
&  +\frac{1}{2}e^{-\kappa\tau}\int_{1}^{t}\kappa\partial_{k}f_{1,1}%
(k,s-1)\hat{Q}_{1}(k,s)ds\nonumber\\
&  -i\frac{\kappa}{2}\left(  1-\frac{1}{\sqrt{-ik}t}\right)  te^{-\sqrt{-ik}%
t}\int_{1}^{\infty}(s-1)\hat{Q}_{1}(0,s)ds~.\nonumber
\end{align}
We rewrite this expression as a sum of terms which can easily be
bounded.\ Namely,%
\[
\kappa\partial_{k}\hat{\omega}_{1,1,1}(k,t)+\kappa\partial_{k}\hat{\omega
}_{2,1,1}(k,t)-\kappa\partial_{k}\hat{\omega}_{\mathrm{as}}(k,t)=\sum
_{i=1}^{5}\kappa\partial_{k}\hat{\omega}_{i}^{r}~,
\]
with%
\begin{align*}
\kappa\partial_{k}\hat{\omega}_{1}^{r}  &  =-\frac{2k-i}{4}\tau\left(
e^{-\kappa(t-1)}-e^{-\kappa t}\right)  \int_{1}^{t}f_{1,1}(k,s-1)\hat{Q}%
_{1}(k,s)ds~,\\
\kappa\partial_{k}\hat{\omega}_{2}^{r}  &  =-\frac{2k-i}{4}\tau e^{-\kappa
t}\int_{1}^{t}f_{1,1}(k,s-1)\hat{Q}_{1}(k,s)ds-i\frac{\kappa}{2}%
te^{-\sqrt{-ik}t}\int_{1}^{t}(s-1)\hat{Q}_{1}(0,s)ds~,\\
\kappa\partial_{k}\hat{\omega}_{3}^{r}  &  =\frac{1}{2}\left(  e^{-\kappa
(t-1)}-e^{-\kappa t}\right)  \int_{1}^{t}\kappa\partial_{k}f_{1,1}%
(k,s-1)\hat{Q}_{1}(k,s)ds~,\\
\kappa\partial_{k}\hat{\omega}_{4}^{r}  &  =\frac{1}{2}e^{-\kappa t}\int
_{1}^{t}\kappa\partial_{k}f_{1,1}(k,s-1)\hat{Q}_{1}(k,s)ds+\frac{i\kappa
}{2\sqrt{-ik}}e^{-\sqrt{-ik}t}\int_{1}^{t}(s-1)\hat{Q}_{1}(0,s)ds~,\\
\kappa\partial_{k}\hat{\omega}_{5}^{r}  &  =-i\frac{\kappa}{2}\left(
1-\frac{1}{\sqrt{-ik}t}\right)  te^{-\sqrt{-ik}t}\int_{t}^{\infty}(s-1)\hat
{Q}_{1}(0,s)ds~.
\end{align*}
In the rest of this proof, we apply without mention (\ref{eq:kfortinmu}) to
eliminate spurious powers of $|\Lambda_{-}|$ whenever the conditions of
Propositions~\ref{prop:sgL1} and \ref{prop:sgL2} require it.

First we have%
\begin{align*}
|\kappa\partial_{k}\hat{\omega}_{1}^{r}|  &  =\left\vert \frac{2k-i}{4}%
\tau\left(  e^{-\kappa(t-1)}-e^{-\kappa t}\right)  \int_{1}^{t}f_{1,1}%
(k,s-1)\hat{Q}_{1}(k,s)ds\right\vert \\
&  \leq\mathrm{const.}~(1+|k|)te^{\Lambda_{-}(t-1)}|\Lambda_{-}|\int_{1}%
^{t}(1+|\Lambda_{-}|)e^{|\Lambda_{-}|\sigma}\min\{1,|\Lambda_{-}|\sigma
\}\mu_{1}(k,s)ds\\
&  \leq\mathrm{const.}~t\left(  \frac{1}{t^{2}}\tilde{\mu}_{\alpha-2}+\frac
{1}{t^{7/2}}\bar{\mu}_{\alpha-2}+\frac{1}{t^{4}}\tilde{\mu}_{\alpha-2}\right)
~,
\end{align*}
showing that $\kappa\partial_{k}\hat{\omega}_{1}^{r}\in\mathcal{B}%
_{\alpha^{\prime\prime},\frac{5}{2},1}$.

For $\kappa\partial_{k}\hat{\omega}_{2}^{r}$ we have%
\begin{align}
&  -\frac{2k-i}{4}\tau e^{-\kappa t}\int_{1}^{t}f_{1,1}(k,s-1)\hat{Q}%
_{1}(k,s)ds-i\frac{\kappa}{2}te^{-\sqrt{-ik}t}\int_{1}^{t}(s-1)\hat{Q}%
_{1}(0,s)ds\nonumber\\
&  =-\frac{2k-i}{4}te^{-\kappa t}\int_{1}^{t}f_{1,1}(k,s-1)\hat{Q}%
_{1}(k,s)ds-i\frac{\kappa}{2}te^{-\sqrt{-ik}t}\int_{1}^{t}(s-1)\hat{Q}%
_{1}(0,s)ds\label{eq:dkwr2a}\\
&  +e^{-\kappa t}\int_{1}^{t}\frac{2k-i}{4}f_{1,1}(k,s-1)\hat{Q}%
_{1}(k,s)ds~,\nonumber
\end{align}
where the last term can be bounded by applying Propositions~\ref{prop:sgL1}
and \ref{prop:sgL2}, so that%
\begin{align*}
&  \left\vert e^{-\kappa t}\int_{1}^{t}\frac{2k-i}{4}f_{1,1}(k,s-1)\hat{Q}%
_{1}(k,s)ds\right\vert \\
&  \leq\mathrm{const.}~e^{\Lambda_{-}\tau}\int_{1}^{t}(1+|\Lambda_{-}%
|^{2})e^{|\Lambda_{-}|\sigma}\min\{1,|\Lambda_{-}|\sigma\}\mu_{1}(k,s)ds\\
&  \leq\mathrm{const.}\left(  \frac{1}{t}\tilde{\mu}_{\alpha-1}(k,t)+\frac
{1}{t^{5/2}}\bar{\mu}_{\alpha-1}(k,t)+\frac{1}{t^{3}}\tilde{\mu}_{\alpha
-1}(k,t)\right)  ~,
\end{align*}
whereas for (\ref{eq:dkwr2a}), we get%
\begin{align}
&  -\frac{2k-i}{4}te^{-\kappa t}\int_{1}^{t}f_{1,1}(k,s-1)\hat{Q}%
_{1}(k,s)ds-i\frac{\kappa}{2}te^{-\sqrt{-ik}t}\int_{1}^{t}(s-1)\hat{Q}%
_{1}(0,s)ds\nonumber\\
&  =-\frac{t}{2}\left(  e^{-\kappa t}\int_{1}^{t}\frac{2k-i}{2}f_{1,1}%
(k,s-1)\hat{Q}_{1}(k,s)ds+e^{-\sqrt{-ik}t}\int_{1}^{t}i\kappa(s-1)\hat{Q}%
_{1}(0,s)ds\right) \nonumber\\
&  =\frac{t}{2}\left(  e^{-\kappa t}-e^{-\sqrt{-ik}t}\right)  \int_{1}%
^{t}i\kappa(s-1)\hat{Q}_{1}(0,s)ds\label{eq:dkwr2b}\\
&  -\frac{t}{2}e^{-\kappa t}\int_{1}^{t}\left(  \frac{2k-i}{2}f_{1,1}%
(k,s-1)\hat{Q}_{1}(k,s)+i\kappa(s-1)\hat{Q}_{1}(0,s)\right)  ds~.
\label{eq:dkwr2c}%
\end{align}
For (\ref{eq:dkwr2b}) we get, using (\ref{eq:exp(sqrt(-ik)-Kappa)}) and
(\ref{eq:kqexpsquareikt}),%
\[
\left\vert \frac{t}{2}\left(  e^{-\kappa t}-e^{-\sqrt{-ik}t}\right)  \int
_{1}^{t}i\kappa(s-1)\hat{Q}_{1}(0,s)ds\right\vert \leq\mathrm{const.}%
~t\left\vert e^{-\sqrt{-ik}t}\right\vert |k|^{3/2}t(|k|^{1/2}+|k|)\in
\mathcal{B}_{\alpha,\infty,2}~.
\]
To bound (\ref{eq:dkwr2c}) we note that, using (\ref{eq:meanValueTheoremQ1}),%
\begin{align*}
&  \frac{2k-i}{2}f_{1,1}(k,s-1)\hat{Q}_{1}(k,s)+i\kappa(s-1)\hat{Q}_{1}(0,s)\\
&  =\left(  \frac{2k-i}{2}f_{1,1}(k,s-1)+i\kappa(s-1)\right)  \hat{Q}%
_{1}(k,s)+i\kappa k(s-1)\partial_{k}\hat{Q}_{1}(\zeta,s)~,
\end{align*}
for some $\zeta\in\lbrack0,k]$, which allows us to rewrite (\ref{eq:dkwr2c})
as%
\begin{align}
&  e^{-\kappa t}\int_{1}^{t}\left(  \frac{2k-i}{2}f_{1,1}(k,s-1)\hat{Q}%
_{1}(k,s)+i\kappa(s-1)\hat{Q}_{1}(0,s)\right)  ds\nonumber\\
&  =e^{-\kappa t}\int_{1}^{t}\left(  \frac{2k-i}{2}f_{1,1}(k,s-1)+i\kappa
(s-1)\right)  \hat{Q}_{1}(k,s)ds\label{eq:dkwr2d}\\
&  +e^{-\kappa t}\int_{1}^{t}i\kappa k(s-1)\partial_{k}\hat{Q}_{1}%
(\zeta,s)ds~.\nonumber
\end{align}
For the last term we have, using (\ref{eq:kqexpsquareikt}),%
\[
\left\vert e^{-\kappa t}\int_{1}^{t}i\kappa k(s-1)\partial_{k}\hat{Q}%
_{1}(\zeta,s)ds\right\vert \leq\mathrm{const.}~e^{\Lambda_{-}t}|k|(|k|^{1/2}%
+|k|)\sqrt{t}\in\mathcal{B}_{\alpha,\infty,\frac{5}{2}}~.
\]
To bound (\ref{eq:dkwr2d}) we use that%
\begin{align*}
\frac{2k-i}{2}f_{1,1}(k,\sigma)+i\kappa\sigma &  =\frac{2k-i}{2}\left(
e^{\kappa\sigma}+\frac{(|k|+\kappa)^{2}}{ik}e^{-\kappa\sigma}-2\frac
{|k|(|k|+\kappa)}{ik}e^{-|k|\sigma}\right)  +i\kappa\sigma\\
&  =k\left(  e^{\kappa\sigma}+\frac{(|k|+\kappa)^{2}}{ik}e^{-\kappa\sigma
}-2\frac{|k|(|k|+\kappa)}{ik}e^{-|k|\sigma}\right) \\
&  -\frac{i}{2}\left(  \left(  e^{\kappa\sigma}-1-\kappa\sigma\right)
-\left(  e^{-\kappa\sigma}-1+\kappa\sigma\right)  \right) \\
&  +\frac{|k|(|k|+\kappa)}{k}\left(  (e^{-\kappa\sigma}-1)-(e^{-|k|\sigma
}-1\right)  )~,
\end{align*}
which, using the usual bound on $f_{1,1}$ and using the fact that leading
order terms\ cancel where we put them in evidence, we get%
\begin{align*}
&  \left\vert \frac{2k-i}{2}f_{1,1}(k,\sigma)+i\kappa\sigma\right\vert \\
&  \leq\mathrm{const.}~(|k|(1+|\Lambda_{-}|)\min\{1,|\Lambda_{-}|\sigma
\}+\min\{1,|\Lambda_{-}|\sigma\}|\Lambda_{-}|\sigma+|\Lambda_{-}|^{2}%
\sigma)e^{|\Lambda_{-}|\sigma}~,
\end{align*}
which, using Propositions~\ref{prop:sgL1} and \ref{prop:sgL2}, yields%
\[
e^{-\kappa t}\int_{1}^{t}\left(  \frac{2k-i}{2}f_{1,1}(k,s-1)+i\kappa
(s-1)\right)  \hat{Q}_{1}(k,s)ds\in\mathcal{B}_{\alpha-1,\frac{5}{2},2}~.
\]
All in all, we thus have $\kappa\partial_{k}\hat{\omega}_{2}^{r}\in
\mathcal{B}_{\alpha^{\prime},\frac{5}{2},1}$.

To bound $\kappa\partial_{k}\hat{\omega}_{3}^{r}$ we use the bound (see
\cite{Boeckle.Wittwer-Decayestimatessolutions2011})%
\[
\left\vert \kappa\partial_{k}f_{1,1}(k,\sigma)\right\vert \leq\mathrm{const.}%
~(1+|\Lambda_{-}|^{2})\sigma e^{|\Lambda_{-}|\sigma}~,
\]
and using Propositions~\ref{prop:sgL1} and \ref{prop:sgL2} we get%
\begin{align*}
&  \left\vert \frac{1}{2}\left(  e^{-\kappa(t-1)}-e^{-\kappa t}\right)
\int_{1}^{t}\kappa\partial_{k}f_{1,1}(k,s-1)\hat{Q}_{1}(k,s)ds\right\vert \\
&  \leq\mathrm{const.}~e^{\Lambda_{-}(t-1)}|\Lambda_{-}|\int_{1}%
^{t}(1+|\Lambda_{-}|^{2})\sigma e^{|\Lambda_{-}|\sigma}\mu_{1}(k,s)ds\\
&  \leq\mathrm{const.}\left(  \frac{1}{t^{1}}\tilde{\mu}_{\alpha-1}+\frac
{1}{t^{5/2}}\bar{\mu}_{\alpha-2}-\frac{1}{t^{3}}\bar{\mu}_{\alpha-2}\right)
~.
\end{align*}
Thus, $\kappa\partial_{k}\hat{\omega}_{3}^{r}\in\mathcal{B}_{\alpha
^{\prime\prime},\frac{5}{2},1}$.

To bound $\kappa\partial_{k}\hat{\omega}_{4}^{r}$ we use that
\begin{align}
\kappa\partial_{k}\hat{\omega}_{4}^{r}  &  =\frac{1}{2}e^{-\kappa t}\int
_{1}^{t}\kappa\partial_{k}f_{1,1}(k,s-1)\hat{Q}_{1}(k,s)ds-\frac{1}%
{2}e^{-\sqrt{-ik}t}\int_{1}^{t}\frac{\kappa\sqrt{-ik}}{k}(s-1)\hat{Q}%
_{1}(0,s)ds\nonumber\\
&  =\frac{1}{2}\left(  e^{-\kappa t}-e^{-\sqrt{-ik}t}\right)  \int_{1}%
^{t}\frac{\kappa\sqrt{-ik}}{k}(s-1)\hat{Q}_{1}(0,s)ds\label{eq:dkwr4a}\\
&  +\frac{1}{2}e^{-\kappa t}\int_{1}^{t}\left(  \kappa\partial_{k}%
f_{1,1}(k,s-1)\hat{Q}_{1}(k,s)-\frac{\kappa\sqrt{-ik}}{k}(s-1)\hat{Q}%
_{1}(0,s)\right)  ds~. \label{eq:dkwr4b}%
\end{align}
We first bound (\ref{eq:dkwr4a}) using\ (\ref{eq:exp(sqrt(-ik)-Kappa)}) and
(\ref{eq:kqexpsquareikt}). We have%
\begin{align*}
&  \left\vert \frac{1}{2}\left(  e^{-\kappa t}-e^{-\sqrt{-ik}t}\right)
\int_{1}^{t}\frac{\kappa\sqrt{-ik}}{k}(s-1)\hat{Q}_{1}(0,s)ds\right\vert \\
&  \leq\mathrm{const.}\left\vert e^{-\sqrt{-ik}t}\right\vert |k|^{3/2}%
t\frac{(|k|^{1/2}+|k|)|k|^{1/2}}{|k|}\in\mathcal{B}_{\alpha,\infty,2}~.
\end{align*}
To bound (\ref{eq:dkwr4b}) we note that, using (\ref{eq:meanValueTheoremQ1}),%
\begin{align*}
&  \kappa\partial_{k}f_{1,1}(k,s-1)\hat{Q}_{1}(k,s)-\frac{\kappa\sqrt{-ik}}%
{k}(s-1)\hat{Q}_{1}(0,s)\\
&  =\left(  \kappa\partial_{k}f_{1,1}(k,s-1)-\frac{\kappa\sqrt{-ik}}%
{k}(s-1)\right)  \hat{Q}_{1}(k,s)+\frac{\kappa\sqrt{-ik}}{k}k(s-1)\partial
_{k}\hat{Q}_{1}(\zeta,s)~,
\end{align*}
for some $\zeta\in\lbrack0,k]$. We next analyze%
\begin{align*}
\kappa\partial_{k}f_{1,1}(k,\sigma)-\frac{\kappa\sqrt{-ik}}{k}\sigma &
=i\frac{\left(  |k|+\kappa\right)  ^{2}}{|k|}\left(  \left(  e^{-|k|\sigma
}-1\right)  -\left(  e^{-\kappa\sigma}-1\right)  \right) \\
&  +\frac{k^{2}+\kappa^{2}}{2k}\left(  \left(  e^{\kappa\sigma}-1\right)
+\left(  e^{-\kappa\sigma}-1\right)  \right)  \sigma\\
&  +2i\frac{k^{2}+|k|\kappa}{k^{2}}\left(  \frac{k^{2}+\kappa^{2}}%
{2}e^{-\kappa\sigma}-|k|\kappa e^{-|k|\sigma}\right)  \sigma\\
&  +\frac{k^{2}+\kappa^{2}}{k}\sigma-\frac{\kappa\sqrt{-ik}}{k}\sigma~.
\end{align*}
For the last line we have, using (\ref{eq:kappa-sqrt(-ik)Bound}),%
\begin{align*}
\left\vert \frac{k^{2}+\kappa^{2}}{k}-\frac{\kappa\sqrt{-ik}}{k}\right\vert
&  \leq|k|+\mathrm{const.}\frac{|\mathbb{\Lambda}_{-}|}{|k|}\min\{|\Lambda
_{-}|^{2},|\Lambda_{-}|^{3}\}\\
&  \leq\mathrm{const.}~(|k|+|k|^{3})\leq\mathrm{const.}~|\Lambda_{-}%
|^{2}(1+|\Lambda_{-}|)~,
\end{align*}
and therefore
\[
\left\vert \kappa\partial_{k}f_{1,1}(k,\sigma)-\frac{\kappa\sqrt{-ik}}%
{k}\sigma\right\vert \leq\mathrm{const.}~(1+|\Lambda_{-}|^{2})|\Lambda
_{-}|\sigma(\sigma+1)e^{|\Lambda_{-}|\sigma}~.
\]
We can now bound (\ref{eq:dkwr4b}). Namely, we have,%
\begin{align*}
&  \left\vert \frac{1}{2}e^{-\kappa t}\int_{1}^{t}\left(  \kappa\partial
_{k}f_{1,1}(k,s-1)\hat{Q}_{1}(k,s)-\frac{\kappa\sqrt{-ik}}{k}(s-1)\hat{Q}%
_{1}(0,s)\right)  ds\right\vert \\
&  =\left\vert \frac{1}{2}e^{-\kappa t}\int_{1}^{t}\left(  \kappa\partial
_{k}f_{1,1}(k,s-1)-\frac{\kappa\sqrt{-ik}}{k}(s-1)\right)  \hat{Q}%
_{1}(k,s)ds\right\vert \\
&  +\left\vert \frac{1}{2}e^{-\kappa t}\int_{1}^{t}\kappa\sqrt{-ik}%
(s-1)\partial_{k}\hat{Q}_{1}(\zeta,s)ds\right\vert \\
&  \leq\mathrm{const.}~e^{\Lambda_{-}t}\int_{1}^{t}(1+|\Lambda_{-}%
|^{2})|\Lambda_{-}|\sigma se^{|\Lambda_{-}|\sigma}\mu_{1}(k,s)ds\\
&  +\mathrm{const.}~e^{\Lambda_{-}t}(|k|^{1/2}+|k|)|k|^{1/2}\sqrt{t}~,
\end{align*}
where by Propositions~\ref{prop:sgL1} and \ref{prop:sgL2}, and inequality
(\ref{eq:kfortinmu}), the first term is in $\mathcal{B}_{\alpha-2,\frac{3}%
{2},1}$ and where due to (\ref{eq:kqexpsquareikt}) the second term is in
$\mathcal{B}_{\alpha,\infty,\frac{3}{2}}$. Thus we get $\kappa\partial_{k}%
\hat{\omega}_{4}^{r}\in\mathcal{B}_{\alpha^{\prime\prime},\frac{3}{2},1}$.

Finally, to bound $\kappa\partial_{k}\hat{\omega}_{5}^{r}$, we use
(\ref{eq:kqexpsquareikt}), so that
\begin{align*}
&  \left\vert i\frac{\kappa}{2}\left(  1-\frac{1}{\sqrt{-ik}t}\right)
te^{-\sqrt{-ik}t}\int_{t}^{\infty}(s-1)\hat{Q}_{1}(0,s)ds\right\vert \\
&  \leq\mathrm{const.}(|k|^{1/2}t+1)(1+|k|^{1/2})\left\vert e^{-\sqrt{-ik}%
t}\right\vert \frac{1}{t^{3/2}}\in\mathcal{B}_{\alpha,\infty,3/2}~.
\end{align*}

Gathering all the bounds on the $\kappa\partial_{k}\hat{\omega}_{i}^{r}$ leads
to (\ref{eq:dkomegaRemainderSpace}).
\end{proof}

\subsection{Improvement of the bounds on the non-linear
terms\label{sec:finalBoundOpt}}

\subsubsection*{Improvement of the bounds on $\hat{Q}_{0}$ and $\hat{Q}_{1}%
$\label{sec:impQ0Q1}}

From Section~\ref{sec:funeq} we know that%
\[
\hat{Q}_{1}=(\hat{\omega}+\hat{\psi})\ast\hat{\omega}+\hat{F}_{2}%
\in\mathcal{B}_{\alpha,\frac{7}{2},4}~.
\]
The force term $\hat{F}_{2}$ is a function of rapid decrease in $k$ and of
compact support in $t$ and will thus not intervene in our bounds. Using
(\ref{eq:omegaas}), (\ref{eq:omegaRemainderSpace}), (\ref{eq:psias}),
(\ref{eq:psiRemainderSpace}), Propositions~\ref{prop:BapqSpaces} and
\ref{prop:optConvBapq} we have%
\begin{align}
\hat{\omega}\ast\hat{\omega}  &  \in\mathcal{B}_{\alpha,\frac{11}{2}%
,4}~,\label{eq:wconvw}\\
\hat{\psi}\ast\left(  \hat{\omega}-\hat{\omega}_{\mathrm{as,1}}\right)   &
\in\mathcal{B}_{\alpha^{\prime},4,\infty}~,\label{eq:pconvwr}\\
(\hat{\psi}-\hat{\psi}_{\mathrm{as,1}})\ast\hat{\omega}_{\mathrm{as,1}}  &
\in\mathcal{B}_{\alpha^{\prime},4,\infty}~. \label{eq:prconvwa}%
\end{align}
For the term $\hat{\psi}_{\mathrm{as,1}}\ast\hat{\omega}_{\mathrm{as,1}}$ we
can take advantage of the particular form of the explicit functions in direct
space in order to improve the index $p$ by $1/2$ in comparison to what would
be possible with the bounds on the convolution. In direct space we have,%
\[
Q_{1}^{\mathrm{d}}=\mathcal{F}^{-1}[\hat{Q}_{1}^{\mathrm{d}}]=\mathcal{F}%
^{-1}[\hat{\psi}_{\mathrm{as,1}}\ast\hat{\omega}_{\mathrm{as,1}}]=\frac
{1}{y^{3/2}}\psi_{1}(x/y)\cdot\frac{1}{y^{3}}\omega_{W}(x/y^{2})~,
\]
where $\psi_{1}$ and $\omega_{W}$ are explicitly represented by
(\ref{eq:asphi1}) and (\ref{eq:asomegaW}) in Appendix~\ref{sec:explicitasfun}.
We use various properties of these functions as well as their derivatives of
order $n$, represented by the superscript $^{(n)}$, which are easily
understood from their explicit representation and shall thus not be proved. We
show that using the definition of the function spaces $\mathcal{B}%
_{\alpha,\infty,q}$ we can improve the bound on $\hat{Q}_{1}$. We require that
all the terms of the form%
\[
\left\vert (|k|y^{2})^{a}\hat{Q}_{1}^{\mathrm{d}}(k,y)\right\vert
=(|k|y^{2})^{a}\left\vert \int_{\mathbb{R}}e^{ikx}\frac{1}{y^{9/2}}\psi
_{1}(x/y)\omega_{W}(x/y^{2})dx\right\vert ~,
\]
for $a\in\mathbb{N}$, $0\leq a\leq\left\lfloor \alpha\right\rfloor +1$, be
bounded. Since, for $n\geq0$, all the $\psi_{1}^{(n)}$ and $\omega_{W}^{(n)}$
are in $C^{\infty}(\mathbb{R)}$ and vanish for $|z|\rightarrow\infty$, we may
integrate by parts and we have%
\[
\left\vert (|k|y^{2})^{a}\hat{Q}_{1}^{\mathrm{d}}(k,y)\right\vert
=y^{2a}\left\vert \int_{\mathbb{R}}e^{ikx}\frac{1}{y^{9/2}}\partial_{x}%
^{a}\left(  \psi_{1}(x/y)\omega_{W}(x/y^{2})\right)  dx\right\vert ~.
\]
We then make use of the Newton binomial to expand the partial derivative of a
product of functions in terms of a product of ordinary derivatives,%
\begin{align*}
&  \left\vert (|k|y^{2})^{a}\hat{Q}_{1}^{\mathrm{d}}(k,y)\right\vert \\
&  \leq\mathrm{const.}~y^{2a}\int_{\mathbb{R}}\frac{1}{y^{9/2}}\sum_{n=0}%
^{a}\binom{a}{n}\frac{1}{y^{n}}|\hat{\psi}_{1}^{(n)}(x/y)|\frac{1}{y^{2(a-n)}%
}|\omega_{W}^{(a-n)}(x/y^{2})|dx~.
\end{align*}
Using the essential fact that
\begin{equation}
\sup_{z\in\mathbb{R}}\left\{  |z|^{n+3/2}|\psi_{1}^{(n)}(z)|\right\}
=\mathrm{const.}<\infty~,~n\geq0~, \label{eq:suppsidom}%
\end{equation}
and that all the $\omega_{W}^{(n)}$ are zero for $z<0$, we have%
\[
\left\vert (|k|y^{2})^{a}\hat{Q}_{1}^{\mathrm{d}}(k,y)\right\vert
\leq\mathrm{const.}\sum_{n=0}^{a}\frac{1}{y^{9/2-n}}\int_{0}^{\infty
}\left\vert \frac{y^{n+3/2}}{x^{n+3/2}}\omega_{W}^{(a-n)}(x/y^{2})\right\vert
dx~.
\]
Finally, using the change of variables $z=x/y^{2}$ and the crucial fact that
all the $\omega_{W}^{(n)}$ have exponential decay when $z\rightarrow0$, we
have%
\begin{align*}
\left\vert (|k|y^{2})^{a}\hat{Q}_{1}^{\mathrm{d}}(k,y)\right\vert  &
\leq\mathrm{const.}\sum_{n=0}^{a}\frac{y^{n+3/2}}{y^{9/2-n}}\int_{0}^{\infty
}\left\vert \frac{\omega_{W}^{(a-n)}(z)}{y^{2n+3}z^{n+3/2}}\right\vert
y^{2}dz\\
&  \leq\mathrm{const.}\sum_{n=0}^{a}\frac{1}{y^{4}}\int_{0}^{\infty}\left\vert
\frac{\omega_{W}^{(a-n)}(z)}{z^{n+3/2}}\right\vert dz\leq\mathrm{const.}%
\frac{1}{y^{4}}~.
\end{align*}
From this we have $q=4$ and thus $\hat{\psi}_{\mathrm{as,1}}\ast\hat{\omega
}_{\mathrm{as,1}}\in\mathcal{B}_{\alpha,\infty,4}$. We conclude, with
(\ref{eq:wconvw})--(\ref{eq:prconvwa}), that%
\begin{equation}
\hat{Q}_{1}\in\mathcal{B}_{\alpha-1,4,4}~. \label{eq:newQ1bound}%
\end{equation}
Similarly, we have%
\begin{equation}
\hat{Q}_{0}\in\mathcal{B}_{\alpha-1,4,3}~, \label{eq:newQ0bound}%
\end{equation}
where the index $q=3$ is due to the product $\hat{\omega}\ast\hat{\eta}$. In
light of (\ref{eq:newQ1bound}) we define%
\begin{equation}
\mu_{1}^{I}:=\frac{1}{s^{4}}\bar{\mu}_{\alpha^{\prime}}(k,s)+\frac{1}{s^{4}%
}\tilde{\mu}_{\alpha^{\prime}}(k,s)~, \label{eq:defmu1I}%
\end{equation}
to replace (\ref{eq:defmu1}) from now on.

\subsubsection{New bounds}

It is now possible to reevaluate the bounds on all functions presented in
Proposition~\ref{prop:BapqSpaces}.

\begin{proposition}
\label{prop:newBapqSpaces}Let $\alpha^{\prime}>1$ and $\delta>0$. We have
\begin{align*}%
\begin{tabular}
[c]{ll}%
$\hat{\psi}_{1,0}\in\mathcal{B}_{\alpha^{\prime},\frac{3}{2}-\delta,2}$ &
$\hat{\psi}_{1,1}\in\mathcal{B}_{\alpha^{\prime},\frac{1}{2},3}$\\
$\hat{\psi}_{2,0}\in\mathcal{B}_{\alpha^{\prime},3,2}$ & $\hat{\psi}_{2,1}%
\in\mathcal{B}_{\alpha^{\prime},3,3}$\\
$\hat{\psi}_{3,0}\in\mathcal{B}_{\alpha^{\prime},3,2}$ & $\hat{\psi}_{3,1}%
\in\mathcal{B}_{\alpha^{\prime},3,3}$%
\end{tabular}
\ \  &
\begin{tabular}
[c]{ll}%
$\hat{\varphi}_{1,0}\in\mathcal{B}_{\alpha^{\prime},\frac{3}{2}-\delta,2}$ &
$\hat{\varphi}_{1,1}\in\mathcal{B}_{\alpha^{\prime},\frac{1}{2},3}$\\
$\hat{\varphi}_{2,0}\in\mathcal{B}_{\alpha^{\prime},3,2}$ & $\hat{\varphi
}_{2,1}\in\mathcal{B}_{\alpha^{\prime},3,3}$\\
$\hat{\varphi}_{3,0}\in\mathcal{B}_{\alpha^{\prime},3,2}$ & $\hat{\varphi
}_{3,1}\in\mathcal{B}_{\alpha^{\prime},3,3}$%
\end{tabular}
\\%
\begin{tabular}
[c]{ll}%
$\hat{\omega}_{1,0}\in\mathcal{B}_{\alpha^{\prime},4,3-\delta}$ & $\hat
{\omega}_{1,1}\in\mathcal{B}_{\alpha^{\prime},3,1}$\\
$\hat{\omega}_{2,0}\in\mathcal{B}_{\alpha^{\prime},\infty,3}$ & $\hat{\omega
}_{2,1}\in\mathcal{B}_{\alpha^{\prime},\infty,3}$\\
$\hat{\omega}_{3,0}\in\mathcal{B}_{\alpha^{\prime},4,3}$ & $\hat{\omega}%
_{3,1}\in\mathcal{B}_{\alpha^{\prime},3,3}$%
\end{tabular}
\ \  &
\begin{tabular}
[c]{ll}%
$\hat{\eta}_{1,0}\in\mathcal{B}_{\alpha^{\prime},3,2-\delta}$ & $\hat{\eta
}_{1,1}\in\mathcal{B}_{\alpha^{\prime},2,0}$\\
$\hat{\eta}_{2,0}\in\mathcal{B}_{\alpha^{\prime},\infty,2}$ & $\hat{\eta
}_{2,1}\in\mathcal{B}_{\alpha^{\prime},\infty,3}$\\
$\hat{\eta}_{3,0}\in\mathcal{B}_{\alpha^{\prime},3,2}$ & $\hat{\eta}_{3,1}%
\in\mathcal{B}_{\alpha^{\prime},2,2}$%
\end{tabular}
\end{align*}

\end{proposition}

\begin{proof}
This is straightforward by the new bounds (\ref{eq:newQ0bound}) and
(\ref{eq:newQ1bound}). For $\hat{\omega}_{2,1}$ and $\hat{\eta}_{2,1}$ we make
use of an existing factor $e^{\Lambda_{-}(t-1)}$ (see
\cite{Hillairet.Wittwer-Existenceofstationary2009}) and apply
(\ref{eq:mutomutilde}), just as was done for $\hat{\omega}_{2,0}$ and
$\hat{\eta}_{2,0}$ in the proof of Proposition~\ref{prop:BapqSpaces}.
\end{proof}

\begin{remark}
We also have%
\begin{align}
\hat{\omega}-\hat{\omega}_{\mathrm{as,1}}  &  \in\mathcal{B}_{\alpha^{\prime
},3,2}~,\label{eq:newOmegaRemainderSpace}\\
\hat{\eta}-\hat{\eta}_{\mathrm{as,1}}  &  \in\mathcal{B}_{\alpha^{\prime}%
,2,1}~. \label{eq:newEtaRemainderSpace}%
\end{align}

\end{remark}

\subsubsection*{Improvement of the bound on $\partial_{k}\hat{Q}_{1}%
$\label{sec:improvedkQ1}}

From \cite{Boeckle.Wittwer-Decayestimatessolutions2011} we have%
\[
\partial_{k}\hat{Q}_{1}=\hat{v}\ast\partial_{k}\hat{\omega}+\partial_{k}%
\hat{F}_{2}\in\mathcal{B}_{\alpha,\frac{3}{2},2}~.
\]
The term $\partial_{k}\hat{F}_{2}$ is a function of rapid decrease in $k$ and
of compact support in $t$ and will thus not intervene in our bounds We use
Propositions~\ref{prop:newBapqSpaces} and \ref{prop:convdisc},
(\ref{eq:psiRemainderSpace}), (\ref{eq:dkomegaKappaBound}) and
(\ref{eq:dkomegaRemainderSpace}) to show that%
\begin{align}
\hat{\omega}\ast\partial_{k}\hat{\omega}  &  \in\mathcal{B}_{\alpha^{\prime
},4,2}~,\label{eq:wconvdkw}\\
\hat{\psi}\ast(\partial_{k}\hat{\omega}-\partial_{k}\hat{\omega}%
_{\mathrm{as,1}})  &  \in\mathcal{B}_{\alpha^{\prime\prime},\frac{5}{2}%
,\infty}~,\label{eq:pconvdkwr}\\
(\hat{\psi}-\hat{\psi}_{\mathrm{as,1}})\ast\partial_{k}\hat{\omega
}_{\mathrm{as,1}}  &  \in\mathcal{B}_{\alpha^{\prime},2,\infty}~.
\label{eq:prconvdkw}%
\end{align}
Since%
\[
\mathcal{F}^{-1}[\partial_{k}\hat{\omega}_{\mathrm{as}}](x,y)=\frac{x}{y^{3}%
}\omega_{W}(x/y^{2})~,
\]
we again use property (\ref{eq:suppsidom}) of $\psi_{\mathrm{as,1}}$ and the
fact that $\omega_{\mathrm{as,1}}^{(n)}(z<0)=0$, for all $n$, to show that the
convolution product $\hat{\psi}_{\mathrm{as,1}}\ast\partial_{k}\hat{\omega
}_{\mathrm{as}}$ can be bounded in direct space in order to improve the index
$p$ by $1/2$ in comparison to what would be possible with the bounds on
convolution. The calculation is slightly longer than in the previous section,
but the steps are exactly the same, so that we omit the details of the proof
for the sake of concision.
We finally have
\[
\hat{\psi}_{\mathrm{as,1}}\ast\partial_{k}\hat{\omega}_{\mathrm{as}}%
\in\mathcal{B}_{\alpha,\infty,2}~,
\]
and we thus get%
\begin{equation}
\partial_{k}\hat{Q}_{1}\in\mathcal{B}_{\alpha^{\prime\prime},2,2}~.
\label{eq:newdkQ1bound}%
\end{equation}

\subsection{Second order in $\hat{\psi}$ and $\hat{\varphi}$}

Applying the new bound (\ref{eq:newQ1bound}) for $\hat{Q}_{1}$ in a
straightforward manner, and in view of Proposition~\ref{prop:newBapqSpaces}
and Remark~\ref{rem:propBapqSpaces}, we find that the second order terms of
$\hat{\psi}$ and $\hat{\varphi}$ are to be extracted from $\hat{\psi}%
_{1,1}-\hat{\psi}_{\mathrm{as,1}}$ and $\hat{\varphi}_{1,1}-\hat{\varphi
}_{\mathrm{as,1}}$, respectively. Inspecting the limits of these quantities
motivates us, in a similar way as in the case of the leading order of
$\hat{\psi}$ and $\hat{\varphi}$, to define the functions
\begin{align}
\hat{\psi}_{\mathrm{as,2}}(k,t)  &  =-\left(  c_{1}|k|+\frac{1}{2}%
c_{2}ik\right)  e^{-|k|t}~,\label{eq:psias2}\\
\hat{\varphi}_{\mathrm{as,2}}(k,t)  &  =-\left(  c_{1}ik-\frac{1}{2}%
c_{2}|k|\right)  e^{-|k|t}~, \label{eq:phias2}%
\end{align}
with $c_{1}$ as defined by (\ref{eq:c1}) and%
\begin{equation}
c_{2}=\int_{1}^{\infty}\left(  s-1\right)  ^{2}\hat{Q}_{1}\left(  0,s\right)
ds~. \label{eq:c2}%
\end{equation}
Note that $\hat{\psi}_{\mathrm{as,2}},\hat{\varphi}_{\mathrm{as,2}}%
\in\mathcal{B}_{\alpha^{\prime},1,\infty}$. We now show that%
\begin{align}
\hat{\psi}_{1,1}-\hat{\psi}_{\mathrm{as,1}}-\hat{\psi}_{\mathrm{as,2}}  &
\in\mathcal{B}_{\alpha^{\prime\prime},\frac{3}{2}-\delta,\infty}%
~,\label{eq:psi2RemainderSpace}\\
\hat{\varphi}_{1,1}-\hat{\varphi}_{\mathrm{as,1}}-\hat{\varphi}_{\mathrm{as,2}%
}  &  \in\mathcal{B}_{\alpha^{\prime\prime},\frac{3}{2}-\delta,\infty}~.
\label{eq:phi2RemainderSpace}%
\end{align}

\begin{proof}
As already for the leading order term, we have%
\[
\hat{\psi}_{1,1}-\hat{\psi}_{\mathrm{as,1}}=\frac{|k|}{ik}(\hat{\varphi}%
_{1,1}-\hat{\varphi}_{\mathrm{as,1}})~,
\]
so that all bounds for $\hat{\psi}-\hat{\psi}_{\mathrm{as,1}}$ are the same as
the ones for $\hat{\varphi}-\hat{\varphi}_{\mathrm{as,1}}$ and we only need to
present the proof for $\hat{\psi}$. We set%
\[
\hat{\psi}_{1,1}(k,t)-\hat{\psi}_{\mathrm{as,1}}(\hat{\psi},t)-\hat{\psi
}_{\mathrm{as,2}}(k,t)=\sum_{i=1}^{4}\hat{\psi}_{i}^{r,2}~,
\]
where%
\begin{align*}
\hat{\psi}_{1}^{r,2}  &  =\frac{1}{2}\left(  e^{-|k|(t-1)}-e^{-|k|t}\right)
\int_{1}^{t}\left(  h_{1,1}(k,\sigma)+2\kappa\sigma\right)  \hat{Q}%
_{1}(k,s)ds~,\\
\hat{\psi}_{2}^{r,2}  &  =\frac{1}{2}e^{-|k|t}\int_{1}^{t}\left(  \left(
h_{1,1}(k,\sigma)+2\kappa\sigma\right)  \hat{Q}_{1}(k,s)+\left(
2|k|+ik\sigma\right)  \sigma\hat{Q}_{1}(0,s)\right)  ds~,\\
\hat{\psi}_{3}^{r,2}  &  =e^{-|k|t}\int_{1}^{t}\sqrt{-ik}\sigma\hat{Q}%
_{1}(0,s)ds-e^{-|k|(t-1)}\int_{1}^{t}\kappa\sigma\hat{Q}_{1}(k,s)ds~,\\
\hat{\psi}_{4}^{r,2}  &  =\frac{1}{2}e^{-|k|t}\int_{t}^{\infty}\left(
2\sqrt{-ik}+2|k|+ik\sigma\right)  \sigma\hat{Q}_{1}(0,s)ds~.
\end{align*}

We first derive some bounds on $h_{1,1}$, given by (\ref{eq:h11}). One has the
straightforward bound%
\begin{equation}
\left\vert h_{1,1}(k,\sigma)+2\kappa\sigma\right\vert \leq\mathrm{const.}%
~(1+|k|+(|k|^{1/2}+|k|)\sigma)e^{|k|\sigma}~, \label{eq:psi2Bound1stOrderA}%
\end{equation}
and since the leading order terms cancel, we also have%
\begin{align}
&  \left\vert h_{1,1}(k,\sigma)+2\kappa\sigma\right\vert \nonumber\\
&  \leq\left\vert (1-e^{|k|\sigma})+\frac{(|k|+\kappa)^{2}}{ik}(e^{-|k|\sigma
}-1)-2\frac{\kappa(|k|+\kappa)}{ik}(e^{-\kappa\sigma}-1)+2\kappa
\sigma\right\vert \nonumber\\
&  \leq\left\vert (1-e^{|k|\sigma})+\frac{(|k|+\kappa)^{2}}{ik}(e^{-|k|\sigma
}-1)-2\frac{\kappa|k|+|k|^{2}}{ik}(e^{-\kappa\sigma}-1)+2(e^{-\kappa\sigma
}-1)+2\kappa\sigma\right\vert \nonumber\\
&  \leq\mathrm{const.}~(|k|\sigma+(1+|k|)|k|\sigma+(|k|^{1/2}+|k|)^{2}%
\sigma+((|k|^{1/2}+|k|)\sigma)^{c+1})e^{|k|\sigma}\nonumber\\
&  \leq\mathrm{const.}~(|k|^{(c+1)/2}+|k|^{2})\sigma(\sigma+1)^{c}%
e^{|k|\sigma}~, \label{eq:psi2Bound1stOrderB}%
\end{align}
with $c=\{0,1\}$ depending on whether we use the $2\kappa\sigma$ term to
cancel an additional term in the last exponential or not. We have another
straightforward bound, namely%
\begin{equation}
|h_{1,1}(k,\sigma)+2\kappa\sigma+2|k|\sigma+ik\sigma^{2}|\leq\mathrm{const.}%
~(1+|k|+|k|^{1/2}\sigma+|k|\sigma(\sigma+1))e^{|k|\sigma}~,
\label{eq:psi2Bound2ndOrderA}%
\end{equation}
and, using that leading order terms cancel, we also have%
\begin{align*}
&  |h_{1,1}(k,\sigma)+2\kappa\sigma+2|k|\sigma+ik\sigma^{2}|\\
&  \leq\left\vert -e^{|k|\sigma}-e^{-|k|\sigma}+2\frac{|k|(|k|+\kappa)}%
{ik}e^{-|k|\sigma}-2\frac{\kappa(|k|+\kappa)}{ik}e^{-\kappa\sigma}%
+2\kappa\sigma+2|k|\sigma+ik\sigma^{2}\right\vert \\
&  \leq\left\vert \vphantom{\frac{1}{2}}-\left(  e^{-|k|\sigma}-1+|k|\sigma
-\frac{1}{2}|k|^{2}\sigma^{2}\right)  -\left(  e^{|k|\sigma}-1-|k|\sigma
-\frac{1}{2}|k|^{2}\sigma^{2}\right)  -|k|^{2}\sigma^{2}\right. \\
&  +\left.  \vphantom{\frac{1}{2}}2\frac{|k|(|k|+\kappa)}{ik}(e^{-|k|\sigma
}-1+|k|\sigma)-2\frac{\kappa(|k|+\kappa)}{ik}(e^{-\kappa\sigma}-1+\kappa
\sigma)+ik\sigma^{2}\right\vert ~.
\end{align*}
Rearranging the terms we get%
\begin{align}
&  |h_{1,1}(k,\sigma)+2\kappa\sigma+2|k|\sigma+ik\sigma^{2}|\nonumber\\
&  \leq\left\vert \vphantom{\frac{1}{2}}-\left(  e^{-|k|\sigma}-1+|k|\sigma
-\frac{1}{2}|k|^{2}\sigma^{2}\right)  -\left(  e^{|k|\sigma}-1-|k|\sigma
-\frac{1}{2}|k|^{2}\sigma^{2}\right)  \right. \nonumber\\
&  +\left.  \vphantom{\frac{1}{2}}2\frac{|k|(|k|+\kappa)}{ik}(e^{-|k|\sigma
}-1+|k|\sigma)-2\frac{|k|(|k|+\kappa)}{ik}(e^{-\kappa\sigma}-1+\kappa
\sigma)\right. \nonumber\\
&  +\left.  \vphantom{\frac{1}{2}}2(e^{-\kappa\sigma}-1+\kappa\sigma
)-\kappa\sigma^{2}\right\vert \nonumber\\
&  \leq\mathrm{const.}~((|k|^{3}\sigma^{3})+(|k|^{1/2}+|k|)(|k|^{2}%
+|\kappa|^{2})\sigma^{2}+(|\kappa|^{3}\sigma^{3}))e^{|k|\sigma}\nonumber\\
&  \leq\mathrm{const.}~(|k|^{3/2}+|k|^{3})\sigma^{2}(\sigma+1)~.
\label{eq:psi2Bound2ndOrderB}%
\end{align}

We now bound the terms $\hat{\psi}_{i}^{r,2}$. Using
Proposition~\ref{prop:sgk1} with the bound (\ref{eq:psi2Bound1stOrderA}) and
Proposition~\ref{prop:sgk2} as with (\ref{eq:psi2Bound1stOrderB}), as well as
inequality (\ref{eq:kfortinmu}) where necessary, we have%
\begin{align*}
|\hat{\psi}_{1}^{r,2}|  &  =\left\vert \frac{1}{2}\left(  e^{-|k|(t-1)}%
-e^{-|k|t}\right)  \int_{1}^{t}(h_{1,1}(k,\sigma)+2\kappa\sigma)\hat{Q}%
_{1}(k,s)ds\right\vert \\
&  \leq\mathrm{const.}~e^{-|k|t}\left\vert e^{|k|}-1\right\vert \int_{1}%
^{t}|h_{1,1}(k,\sigma)+2\kappa\sigma|\mu_{1}^{I}(k,s)ds\\
&  \leq\mathrm{const.}~e^{-|k|(t-1)}|k|\int_{1}^{\frac{t+1}{2}}%
(1+|k|)|k|\sigma se^{|k|\sigma}\mu_{1}^{I}(k,s)ds\\
&  +\mathrm{const.}~e^{-|k|(t-1)}|k|\int_{\frac{t+1}{2}}^{t}(1+|k|+(|k|^{1/2}%
+|k|)\sigma)e^{|k|\sigma}\mu_{1}^{I}(k,s)ds\\
&  \leq\mathrm{const.}\left(  \frac{1}{t^{2}}\bar{\mu}_{\alpha^{\prime}%
-1}(k,t)+\frac{1}{t^{7/2}}\bar{\mu}_{\alpha^{\prime}-1}(k,t)+\frac{1}{t^{4}%
}\tilde{\mu}_{\alpha^{\prime}-1}(k,t)\right)  ~,
\end{align*}
which shows that $\hat{\psi}_{1}^{r,2}\in\mathcal{B}_{\alpha^{\prime\prime
},2,4}$.

For $\hat{\psi}_{2}^{r,2}$ we split the integration interval into two
sub-intervals, $[1,t^{\rho}]$ and $[t^{\rho},t]$, with $0<\rho<1$. We also
rewrite the integral over the first sub-interval using
(\ref{eq:meanValueTheoremQ1}), so that%
\begin{align}
\hat{\psi}_{2}^{r,2}  &  =\frac{1}{2}e^{-|k|t}\int_{1}^{t}\left(  \left(
h_{1,1}(k,\sigma)+2\kappa\sigma\right)  \hat{Q}_{1}(k,s)+\left(
2|k|+ik\sigma\right)  \sigma\hat{Q}_{1}(0,s)\right)  ds\nonumber\\
&  =\frac{1}{2}e^{-|k|t}\int_{1}^{t^{\rho}}\left(  \left(  h_{1,1}%
(k,\sigma)+2\kappa\sigma+\left(  2|k|+ik\sigma\right)  \sigma\right)  \hat
{Q}_{1}(k,s)\right)  ds\label{eq:psi2ProofB1}\\
&  -\frac{1}{2}e^{-|k|t}\int_{1}^{t^{\rho}}\left(  2|k|+ik\sigma\right)
\sigma k\partial_{k}\hat{Q}_{1}(\zeta,s)ds\label{eq:psi2ProofB2}\\
&  +\frac{1}{2}e^{-|k|t}\int_{t^{\rho}}^{t}\left(  \left(  h_{1,1}%
(k,\sigma)+2\kappa\sigma\right)  \hat{Q}_{1}(k,s)+\left(  2|k|+ik\sigma
\right)  \sigma\hat{Q}_{1}(0,s)\right)  ds~. \label{eq:psi2ProofB3}%
\end{align}
For (\ref{eq:psi2ProofB1}) we have, using Proposition~\ref{prop:sgk1} with the
bound (\ref{eq:psi2Bound2ndOrderB}) and Proposition~\ref{prop:sgk2}, with
(\ref{eq:psi2Bound2ndOrderA}),
\begin{align*}
&  \left\vert \frac{1}{2}e^{-|k|t}\int_{1}^{t^{\rho}}\left(  \left(
h_{1,1}(k,\sigma)+2\kappa\sigma+\left(  2|k|+ik\sigma\right)  \sigma\right)
\hat{Q}_{1}(k,s)\right)  ds\right\vert \\
&  \leq\mathrm{const.}e^{-|k|t}\int_{1}^{t}|h_{1,1}(k,\sigma)+2\kappa
\sigma+\left(  2|k|+ik\sigma\right)  \sigma|\mu_{1}^{I}(k,s)ds\\
&  \leq\mathrm{const.}~e^{-|k|t}\int_{1}^{\frac{t+1}{2}}(|k|^{3/2}%
+|k|^{3})\sigma^{2}se^{|k|\sigma}\mu_{1}^{I}(k,s)ds\\
&  +\mathrm{const.}~e^{-|k|t}\int_{\frac{t+1}{2}}^{t}(1+|k|+|k|^{1/2}%
\sigma+|k|\sigma s)e^{|k|\sigma}\mu_{1}^{I}(k,s)ds\\
&  \leq\mathrm{const.}\left(  \frac{1}{t^{3/2-\delta}}\bar{\mu}_{\alpha
^{\prime}}(k,t)+\frac{1}{t^{2}}\bar{\mu}_{\alpha^{\prime}}(k,t)+\frac{1}%
{t^{2}}\tilde{\mu}_{\alpha^{\prime}}(k,t)\right)  \in\mathcal{B}%
_{\alpha^{\prime},\frac{3}{2}-\delta,2}~.
\end{align*}
For (\ref{eq:psi2ProofB2}) we have, using (\ref{eq:kpexpabskt}),%
\begin{align*}
&  \left\vert \frac{1}{2}e^{-|k|t}\int_{1}^{t^{\rho}}\left(  2|k|+ik\sigma
\right)  \sigma k\partial_{k}\hat{Q}_{1}(\zeta,s)ds\right\vert \\
&  \leq\mathrm{const.}~e^{-|k|t}|k|^{2}\int_{1}^{t^{\rho}}s\sigma\frac
{1}{s^{2}}ds\\
&  \leq\mathrm{const.}~e^{-|k|t}|k|^{2}(t^{\rho}+\log(1+t)-1)\in
\mathcal{B}_{\alpha^{\prime},2-\rho,\infty}~.
\end{align*}
For (\ref{eq:psi2ProofB3}), we have%
\begin{align*}
&  \left\vert \frac{1}{2}e^{-|k|t}\int_{t^{\rho}}^{t}\left(  \left(
h_{1,1}(k,\sigma)+2\kappa\sigma\right)  \hat{Q}_{1}(k,s)+\left(
2|k|+ik\sigma\right)  \sigma\hat{Q}_{1}(0,s)\right)  ds\right\vert \\
&  \leq\mathrm{const.}~e^{-|k|t}\int_{t^{\rho}}^{t}\left\vert \left(
h_{1,1}(k,\sigma)+2\kappa\sigma\right)  \hat{Q}_{1}(k,s)\right\vert
ds+\mathrm{const.}~e^{-|k|t}|k|\frac{1}{t^{\rho}}~,
\end{align*}
where the second term is in $\mathcal{B}_{\alpha^{\prime},1+\rho,\infty}$ by
(\ref{eq:kpexpabskt}). We split the remaining integral into two sub-intervals
after setting $\rho\leq1/2$, and make use of (\ref{eq:psi2Bound1stOrderB}) and
(\ref{eq:psi2Bound1stOrderA}), respectively. We get, using
Proposition~\ref{prop:sgk2} to bound the second integral,
\begin{align*}
&  e^{-|k|t}\int_{t^{\rho}}^{t}|(h_{1,1}(k,\sigma)+2\kappa\sigma)\hat{Q}%
_{1}(k,s)|ds\\
&  =e^{-|k|t}\int_{t^{\rho}}^{\frac{t+1}{2}}(|k|^{1/2}+|k|^{2})\sigma
e^{|k|\sigma}\mu_{1}^{I}(k,s)ds+e^{-|k|t}\int_{\frac{t+1}{2}}^{t}%
(1+|k|+(|k|^{1/2}+|k|)\sigma)e^{|k|\sigma}\mu_{1}^{I}(k,s)ds\\
&  \leq\mathrm{const.}~\left(  e^{-|k|t/2}(|k|^{1/2}+|k|^{2})\frac{1}%
{t^{2\rho}}+\frac{1}{t^{5/2}}\bar{\mu}_{\alpha^{\prime}}(k,t)+\frac{1}%
{t^{5/2}}\tilde{\mu}_{\alpha^{\prime}}(k,t)\right)  ~,
\end{align*}
where (\ref{eq:kpexpabskt}) allows to bound the first term, so that this
expression is in $\mathcal{B}_{\alpha^{\prime},2\rho+\frac{1}{2},\frac{5}{2}}%
$. Therefore, if we chose $\rho=1/2$, then $\hat{\psi}_{2}^{r,2}\in
\mathcal{B}_{\alpha^{\prime},\frac{3}{2}-\delta,\frac{5}{2}}$.

To bound $\hat{\psi}_{3}^{r,2}$ we note that%
\begin{align*}
|\hat{\psi}_{3}^{r,2}|  &  =\left\vert e^{-|k|t}\int_{1}^{t}\sqrt{-ik}%
\sigma\hat{Q}_{1}(0,s)ds-e^{-|k|(t-1)}\int_{1}^{t}\kappa\sigma\hat{Q}%
_{1}(k,s)ds\right\vert \\
&  \leq\left\vert \left(  e^{-|k|t}-e^{-|k|(t-1)}\right)  \int_{1}^{t}%
\kappa\sigma\hat{Q}_{1}(k,s)ds\right\vert \\
&  +\left\vert e^{-|k|t}\int_{1}^{t}\left(  \sqrt{-ik}-\kappa\right)
\sigma\hat{Q}_{1}(k,s)ds\right\vert \\
&  +\left\vert e^{-|k|t}\int_{1}^{t}\sqrt{-ik}k\sigma\partial_{k}\hat{Q}%
_{1}(\zeta,s)ds\right\vert \\
&  \leq\mathrm{const.}~e^{-|k|(t-1)}|k|\int_{1}^{t}(|k|^{1/2}+|k|)e^{|k|\sigma
}\sigma\mu_{1}^{I}(k,s)ds\\
&  +\mathrm{const.}~e^{-|k|t}|k|^{3/2}+\mathrm{const.}~e^{-|k|t}|k|^{3/2}%
\log(1+t)~,
\end{align*}
which, using Propositions~\ref{prop:sgk1} and \ref{prop:sgk2},
(\ref{eq:kappa-sqrt(-ik)Bound}), and when necessary (\ref{eq:kfortinmu}) for
the first term and (\ref{eq:kpexpabskt}) for the other two terms, shows that
$\hat{\psi}_{3}^{r,2}\in\mathcal{B}_{\alpha^{\prime\prime},\frac{3}{2}%
-\delta,5}$.

To bound $\hat{\psi}_{4}^{r,2}$ we simply integrate with respect to $s$ and
then apply (\ref{eq:kpexpabskt}),%
\begin{align*}
|\hat{\psi}_{4}^{r,2}|  &  =\left\vert \frac{1}{2}e^{-|k|t}\int_{t}^{\infty
}\left(  2\sqrt{-ik}+2|k|+ik\sigma\right)  \sigma\hat{Q}_{1}(0,s)ds\right\vert
\\
&  \leq\mathrm{const.}~e^{-|k|t}(|k|^{1/2}+|k|)\frac{1}{t^{2}}+e^{-|k|t}%
|k|\frac{1}{t^{1}}\in\mathcal{B}_{\alpha^{\prime},2,\infty}~.
\end{align*}

Gathering the bounds on the $\hat{\psi}_{i}^{r,2}$ yields
(\ref{eq:psi2RemainderSpace}), and by the opening remark of the proof also
(\ref{eq:phi2RemainderSpace}).
\end{proof}

\subsection{Final improvement of the bounds on $\hat{Q}_{0}$, $\hat{Q}_{1}$,
and $\partial_{k}\hat{Q}_{1}$\label{sec:finalImprovements}}

Using Proposition~\ref{prop:newBapqSpaces}, (\ref{eq:omegaas}),
(\ref{eq:newOmegaRemainderSpace}) and (\ref{eq:psi2RemainderSpace}), we get%
\begin{align*}
\hat{\omega}\ast\hat{\omega}  &  \in\mathcal{B}_{\alpha^{\prime},6,4}~,\\
\hat{\psi}\ast\left(  \hat{\omega}-\hat{\omega}_{\mathrm{as,1}}\right)   &
\in\mathcal{B}_{\alpha^{\prime\prime},\frac{9}{2},\infty}~,\\
(\hat{\psi}-\hat{\psi}_{\mathrm{as,1}}-\hat{\psi}_{\mathrm{as,2}})\ast
\hat{\omega}_{\mathrm{as,1}}  &  \in\mathcal{B}_{\alpha^{\prime\prime}%
,\frac{9}{2}-\delta,\infty}~.
\end{align*}
For the term $(\hat{\psi}_{\mathrm{as,1}}+\hat{\psi}_{\mathrm{as,2}})\ast
\hat{\omega}_{\mathrm{as,1}}$ we can proceed exactly as in
Section~\ref{sec:impQ0Q1}, thanks to the fact that
\[
\sup_{z\in\mathbb{R}}\left\{  |z|^{n+2}|\psi_{2}^{(n)}(z)|\right\}
=\mathrm{const.}<\infty~,~n\geq0~.
\]
We conclude that%
\[
(\hat{\psi}_{\mathrm{as,1}}+\hat{\psi}_{\mathrm{as,2}})\ast\hat{\omega
}_{\mathrm{as,1}}\in\mathcal{B}_{\alpha^{\prime},\infty,4}~,
\]
and therefore%
\begin{equation}
\hat{Q}_{1}\in\mathcal{B}_{\alpha^{\prime\prime},\frac{9}{2}-\delta,4}~.
\label{eq:newestQ1bound}%
\end{equation}
Similarly, we have%
\begin{align}
\hat{Q}_{0}  &  \in\mathcal{B}_{\alpha^{\prime\prime},\frac{9}{2}-\delta
,3}~,\label{eq:newestQ0bound}\\
\partial_{k}\hat{Q}_{1}  &  \in\mathcal{B}_{\alpha^{\prime\prime},\frac{5}%
{2}-\delta,2}~. \label{eq:newestdkQ1bound}%
\end{align}
In the light of (\ref{eq:newestQ1bound}), we define%
\begin{equation}
\mu_{1}^{II}:=\frac{1}{s^{9/2-\delta}}\bar{\mu}_{\alpha^{\prime\prime}%
}(k,s)+\frac{1}{s^{4}}\tilde{\mu}_{\alpha^{\prime\prime}}(k,s)
\label{eq:defmu1II}%
\end{equation}
to replace (\ref{eq:defmu1I}) from now on.

\subsection{Second order in $\hat{\eta}$ and $\hat{\omega}$}

Applying the new bound (\ref{eq:newestQ1bound}) for $\hat{Q}_{1}$ in a
straightforward manner, and in view of Proposition~\ref{prop:newBapqSpaces}
and Remark~\ref{rem:propBapqSpaces}, we find that the second order terms of
$\hat{\eta}$ and $\hat{\omega}$ are to be extracted from $\hat{\eta}%
_{1,1}-\hat{\eta}_{\mathrm{as,1}}$ and $\hat{\omega}_{1,1}-\hat{\omega
}_{\mathrm{as,1}}$, respectively. Inspecting the limits of these quantities
motivates us, in a similar way as in the case of the leading order of
$\hat{\eta}$ and $\hat{\omega}$, to define the functions%
\begin{align}
\hat{\eta}_{\mathrm{as,2}}(k,t)  &  =-c_{1}\frac{|k|-ik}{\sqrt{-ik}}%
e^{-\sqrt{-ik}t}~,\label{eq:etaas2}\\
\hat{\omega}_{\mathrm{as,2}}(k,t)  &  =c_{1}(|k|-ik)e^{-\sqrt{-ik}t}~,
\label{eq:omegaas2}%
\end{align}
with $c_{1}$ as defined by (\ref{eq:c1}). Note that $\hat{\eta}_{\mathrm{as,2}%
}\in\mathcal{B}_{\alpha^{\prime},\infty,1}$ and $\hat{\omega}_{\mathrm{as,2}%
}\in\mathcal{B}_{\alpha^{\prime},\infty,2}$. We now show that%
\begin{align}
\hat{\eta}_{1,1}-\hat{\eta}_{\mathrm{as,1}}-\hat{\eta}_{\mathrm{as,2}}  &
\in\mathcal{B}_{\alpha^{\prime\prime}-1,\frac{5}{2}-\delta,2-\delta
}~,\label{eq:eta2RemainderSpace}\\
\hat{\omega}_{1,1}-\hat{\omega}_{\mathrm{as,1}}-\hat{\omega}_{\mathrm{as,2}}
&  \in\mathcal{B}_{\alpha^{\prime\prime}-1,\frac{7}{2}-\delta,3-\delta}~.
\label{eq:omega2RemainderSpace}%
\end{align}

\begin{remark}
The bounds (\ref{eq:newestQ0bound}) and (\ref{eq:newestQ1bound}) for $\hat
{Q}_{0}$ and $\hat{Q}_{1}$, respectively, also show that $\hat{\omega}%
_{2,1}\in\mathcal{B}_{\alpha^{\prime\prime},\infty,3}$, using
(\ref{eq:mutomutilde}) and $\hat{\omega}_{3,1}\in\mathcal{B}_{\alpha
^{\prime\prime},\frac{7}{2}-\delta,3}$. This means, as is already mentioned in
Remark~\ref{rem:propBapqSpaces}, that only $\hat{\omega}_{1,1}$ plays a role
in (\ref{eq:asw}).
\end{remark}

\begin{proof}
As for the leading order term, we have%
\[
\hat{\omega}_{1,1}-\hat{\omega}_{\mathrm{as,1}}=\frac{ik}{\kappa}(\hat{\eta
}_{1,1}-\hat{\eta}_{\mathrm{as,1}})~,
\]
so that for the same reasons, the $\mathcal{B}_{\alpha,p,q}$ space of the
second order term of $\hat{\omega}$ has indices $p$ and $q$ greater by $1$
than that of the second order term of $\hat{\eta}$, and thus we only present
the proof for $\hat{\eta}$.

In order to prove (\ref{eq:eta2RemainderSpace}) we analyze%
\begin{align*}
\hat{\eta}_{1,1}(k,t)-\hat{\eta}_{\mathrm{as,1}}(k,t)-\hat{\eta}%
_{\mathrm{as,2}}(k,t)  &  =\frac{1}{2}e^{-\kappa(t-1)}\int_{1}^{t}%
g_{1,1}(k,s-1)\hat{Q}_{1}(k,s)ds\\
&  +\frac{1}{2}e^{-\sqrt{-ik}t}\int_{1}^{\infty}2\left(  1+\frac{|k|-ik}%
{\sqrt{-ik}}\right)  (s-1)\hat{Q}_{1}(0,s)ds~.
\end{align*}
We rewrite this expression as a sum of terms which can easily be
bounded.\ Namely,%
\[
\hat{\eta}_{1,1}(k,t)-\hat{\eta}_{\mathrm{as,1}}(k,t)-\hat{\eta}%
_{\mathrm{as,2}}(k,t)=\sum_{i=1}^{6}\hat{\eta}_{i}^{r,2}~,
\]
with%
\begin{align*}
\hat{\eta}_{1}^{r,2}  &  =\frac{1}{2}\left(  e^{-\kappa(t-1)}-e^{-\kappa
t}\right)  \int_{1}^{t}\left(  g_{1,1}(k,\sigma)-2\frac{\kappa}{ik}%
\kappa\sigma\right)  \hat{Q}_{1}(k,s)ds~,\\
\hat{\eta}_{2}^{r,2}  &  =\frac{1}{2}e^{-\kappa t}\int_{1}^{t}\left(  \left(
g_{1,1}(k,\sigma)-2\frac{\kappa}{ik}\kappa\sigma\right)  \hat{Q}%
_{1}(k,s)+\frac{2|k|}{\sqrt{-ik}}\sigma\hat{Q}_{1}(0,s)\right)  ds~,\\
\hat{\eta}_{3}^{r,2}  &  =\left(  e^{-\kappa(t-1)}-e^{-\kappa t}\right)
\int_{1}^{t}\frac{\kappa}{ik}\kappa\sigma\hat{Q}_{1}(k,s)ds+e^{-\sqrt{-ik}%
t}\int_{1}^{t}\frac{-ik}{\sqrt{-ik}}\sigma\hat{Q}_{1}(0,s)ds~,\\
\hat{\eta}_{4}^{r,2}  &  =e^{-\kappa t}\int_{1}^{t}\left(  \frac{\kappa}%
{ik}\kappa\hat{Q}_{1}(k,s)+\hat{Q}_{1}(0,s)\right)  \sigma ds~,\\
\hat{\eta}_{5}^{r,2}  &  =-\left(  e^{-\kappa t}-e^{-\sqrt{-ik}t}\right)
\int_{1}^{t}\left(  1+\frac{|k|}{\sqrt{-ik}}\right)  \sigma\hat{Q}%
_{1}(0,s)ds~,\\
\hat{\eta}_{6}^{r,2}  &  =e^{-\sqrt{-ik}t}\int_{t}^{\infty}\left(
1+\frac{|k|-ik}{\sqrt{-ik}}\right)  \sigma\hat{Q}_{1}(0,s)ds~.
\end{align*}

The term $\hat{\eta}_{1}^{r,2}$ must be bounded by
Propositions~\ref{prop:sgL1} and \ref{prop:sgL2} for $|k|\leq1$ and $|k|>1$
separately. We use the bounds%
\[
\left\vert g_{1,1}(k,\sigma)-2\frac{\kappa^{2}}{ik}\sigma\right\vert
\leq\left\{
\begin{array}
[c]{lc}%
\mathrm{const.}~\sigma e^{|\Lambda_{-}|\sigma}\min\{1,|\Lambda_{-}%
|(\sigma+1)\} & \mathrm{for}~|k|\leq1\\
\mathrm{const.}~e^{|\Lambda_{-}|\sigma}\min\{(1+|\Lambda_{-}|s),|\Lambda
_{-}|^{2}\sigma(\sigma+1)\} & \mathrm{for}~|k|>1
\end{array}
\right.  .
\]
which can easily be obtained from (\ref{eq:eta2BoundTot}). For $|k|\leq1$ we
have%
\begin{align*}
|\hat{\eta}_{1}^{r,2}|  &  =\left\vert \frac{1}{2}\left(  e^{-\kappa
(t-1)}-e^{-\kappa t}\right)  \int_{1}^{t}(g_{1,1}(k,\sigma)-2\frac{\kappa}%
{ik}\kappa\sigma)\hat{Q}_{1}(k,s)ds\right\vert \\
&  \leq\mathrm{const.}~e^{\Lambda_{-}(t-1)}|\Lambda_{-}|\int_{1}^{\frac
{t+1}{2}}|\Lambda_{-}|\sigma se^{|\Lambda_{-}|\sigma}\mu_{1}^{II}(k,s)ds\\
&  +\mathrm{const.}~e^{\Lambda_{-}(t-1)}|\Lambda_{-}|\int_{\frac{t+1}{2}}%
^{t}\sigma e^{|\Lambda_{-}|\sigma}\mu_{1}^{II}(k,s)ds\\
&  \leq\mathrm{const.}~\left(  \frac{1}{t^{2}}\tilde{\mu}_{\alpha
^{\prime\prime}}(k,t)+\frac{1}{t^{7/2-\delta}}\bar{\mu}_{\alpha^{\prime\prime
}}(k,t)+\frac{1}{t^{3}}\tilde{\mu}_{\alpha^{\prime\prime}}(k,t)\right)  ~,
\end{align*}
and for $|k|>1$, using (\ref{eq:kfortinmu}) to deal with the spurious
$|\Lambda_{-}|$ factor,%
\begin{align*}
|\hat{\eta}_{1}^{r,2}|  &  =\left\vert \frac{1}{2}\left(  e^{-\kappa
(t-1)}-e^{-\kappa t}\right)  \int_{1}^{t}(g_{1,1}(k,\sigma)-2\frac{\kappa}%
{ik}\kappa\sigma)\hat{Q}_{1}(k,s)ds\right\vert \\
&  \leq\mathrm{const.}~e^{\Lambda_{-}(t-1)}|\Lambda_{-}|\int_{1}^{\frac
{t+1}{2}}|\Lambda_{-}|^{2}\sigma se^{|\Lambda_{-}|\sigma}\mu_{1}^{II}(k,s)ds\\
&  +\mathrm{const.}~e^{\Lambda_{-}(t-1)}|\Lambda_{-}|\int_{\frac{t+1}{2}}%
^{t}(1+|\Lambda_{-}|s)e^{|\Lambda_{-}|\sigma}\mu_{1}^{II}(k,s)ds\\
&  \leq\mathrm{const.}\left(  \frac{1}{t^{4}}\tilde{\mu}_{\alpha^{\prime
\prime}-1}(k,t)+\frac{1}{t^{9/2-\delta}}\bar{\mu}_{\alpha^{\prime\prime}%
-1}(k,t)+\frac{1}{t^{4}}\tilde{\mu}_{\alpha^{\prime\prime}-1}(k,t)\right)  ~.
\end{align*}
This shows that $\hat{\eta}_{1}^{r,2}$ is in $\mathcal{B}_{\alpha
^{\prime\prime}-1,\frac{7}{2}-\delta,2}$.

For $\hat{\eta}_{2}^{r,2}$ we use the fact that, using
(\ref{eq:meanValueTheoremQ1}),%
\begin{align*}
&  \left(  g_{1,1}(k,\sigma)-2\frac{\kappa}{ik}\kappa\sigma\right)  \hat
{Q}_{1}(k,s)+\frac{2|k|}{\sqrt{-ik}}\sigma\hat{Q}_{1}(0,s)\\
&  =\left(  g_{1,1}(k,\sigma)-2\frac{\kappa}{ik}\kappa\sigma+\frac{2|k|}%
{\sqrt{-ik}}\sigma\right)  \hat{Q}_{1}(k,s)-\frac{2|k|}{\sqrt{-ik}}\sigma
k\partial_{k}\hat{Q}_{1}(\zeta,s)~,
\end{align*}
for some $\zeta\in\lbrack0,k]$. We analyze the expression%
\[
g_{1,1}(k,\sigma)-2\frac{\kappa}{ik}\kappa\sigma+\frac{2|k|}{\sqrt{-ik}}%
\sigma=\frac{\kappa}{ik}\left(  e^{\kappa\sigma}+\frac{(|k|+\kappa)^{2}}%
{ik}e^{-\kappa\sigma}-2\frac{|k|(|k|+\kappa)}{ik}e^{-|k|\sigma}-2\kappa
\sigma-\frac{2|k|\sqrt{-ik}}{\kappa}\sigma\right)
\]
in some more detail. A straightforward bound is%
\begin{align}
&  \left\vert g_{1,1}(k,\sigma)-2\frac{\kappa}{ik}\kappa\sigma+\frac
{2|k|}{\sqrt{-ik}}\sigma\right\vert \nonumber\\
&  =\left\vert \frac{\kappa}{ik}\left(  e^{\kappa\sigma}+\frac{(|k|+\kappa
)^{2}}{ik}e^{-\kappa\sigma}-2\frac{|k|(|k|+\kappa)}{ik}e^{-|k|\sigma}%
-2\kappa\sigma-\frac{2|k|\sqrt{-ik}}{\kappa}\right)  \right\vert \nonumber\\
&  \leq\left\vert \frac{\kappa}{ik}\right\vert \left\vert (e^{\kappa\sigma
}-1)-(e^{-\kappa\sigma}-1)+\frac{2|k|^{2}+2|k|\kappa}{ik}\left(
(e^{-\kappa\sigma}-1)-(e^{-|k|\sigma}-1)\right)  \right\vert \nonumber\\
&  +\left\vert \frac{\kappa}{ik}\right\vert \left\vert 2\kappa\sigma
+\frac{2|k|\sqrt{-ik}}{\kappa}\sigma\right\vert \leq\left\{
\begin{array}
[c]{lc}%
\mathrm{const.}~(1+|\Lambda_{-}|)\sigma e^{|\Lambda_{-}|\sigma} &
\mathrm{for}~|k|\leq1\\
\mathrm{const.}~(1+|\Lambda_{-}|(\sigma+1))e^{|\Lambda_{-}|\sigma} &
\mathrm{for}~|k|>1
\end{array}
\right.  . \label{eq:eta2Bound1}%
\end{align}
but we may also cancel leading order terms so that%
\begin{align*}
&  g_{1,1}(k,\sigma)-2\frac{\kappa}{ik}\kappa\sigma+\frac{2|k|}{\sqrt{-ik}%
}\sigma\\
&  =\frac{\kappa}{ik}\left(  (e^{\kappa\sigma}-1-\kappa\sigma-\frac{1}%
{2}\kappa^{2}\sigma^{2})-(e^{-\kappa\sigma}-1+\kappa\sigma-\frac{1}{2}%
\kappa^{2}\sigma^{2})\right) \\
&  +\frac{\kappa}{ik}\left(  \frac{2|k|^{2}+2|k|\kappa}{ik}\left(
(e^{-\kappa\sigma}-1+\kappa\sigma)-(e^{-|k|\sigma}-1+|k|\sigma)\right)
\right) \\
&  +\frac{\kappa}{ik}\left(  \frac{2|k|^{2}+2|k|\kappa}{ik}(|k|-\kappa
)\sigma-\frac{2|k|\sqrt{-ik}}{\kappa}\sigma\right)  ~,
\end{align*}
where the third term reduces to $-2i|k|\left(  \kappa-\sqrt{-ik}\right)
\sigma/k$. This yields%
\begin{equation}
\left\vert g_{1,1}(k,\sigma)-2\frac{\kappa}{ik}\kappa\sigma+\frac{2|k|}%
{\sqrt{-ik}}\sigma\right\vert \leq\left\{
\begin{array}
[c]{lc}%
\mathrm{const.}~|\Lambda_{-}|^{2}\sigma(\sigma^{2}+\sigma+1)e^{|\Lambda
_{-}|\sigma} & \mathrm{for}~|k|\leq1\\
\mathrm{const.}~|\Lambda_{-}|^{3}\sigma\left(  \sigma^{2}+\sigma+1\right)
e^{|\Lambda_{-}|\sigma} & \mathrm{for}~|k|>1
\end{array}
\right.  . \label{eq:eta2Bound2}%
\end{equation}
Collecting (\ref{eq:eta2Bound1}) and (\ref{eq:eta2Bound2}) we have%
\begin{align}
&  \left\vert g_{1,1}(k,\sigma)-2\frac{\kappa}{ik}\kappa\sigma+\frac
{2|k|}{\sqrt{-ik}}\sigma\right\vert \nonumber\\
&  \leq\left\{
\begin{array}
[c]{lc}%
\mathrm{const.}~\min\{(1+|\Lambda_{-}|),|\Lambda_{-}|^{2}(\sigma^{2}%
+\sigma+1)\}\sigma e^{|\Lambda_{-}|\sigma} & \mathrm{for}~|k|\leq1\\
\mathrm{const.}~\min\{(1+|\Lambda_{-}|(\sigma+1)),|\Lambda_{-}|^{3}%
\sigma\left(  \sigma^{2}+\sigma+1\right)  \}e^{|\Lambda_{-}|\sigma} &
\mathrm{for}~|k|>1
\end{array}
\right.  . \label{eq:eta2BoundTot}%
\end{align}
We can now bound $\hat{\eta}_{2}^{r,2}$ by splitting it into two terms. We
have%
\begin{align}
|\hat{\eta}_{2}^{r,2}|  &  =\left\vert \frac{1}{2}e^{-\kappa t}\int_{1}%
^{t}\left(  \left(  g_{1,1}(k,\sigma)-2\frac{\kappa}{ik}\kappa\sigma\right)
\hat{Q}_{1}(k,s)+\frac{2|k|}{\sqrt{-ik}}\sigma\hat{Q}_{1}(0,s)\right)
ds\right\vert \nonumber\\
&  \leq\left\vert \frac{1}{2}e^{-\kappa t}\int_{1}^{t}(g_{1,1}(k,\sigma
)-2\frac{\kappa}{ik}\kappa\sigma+\frac{2|k|}{\sqrt{-ik}}\sigma)\hat{Q}%
_{1}(k,s)ds\right\vert \label{eq:eta2ProofB1}\\
&  +\left\vert e^{-\kappa t}\int_{1}^{t}\frac{|k|}{\sqrt{-ik}}\sigma
k\partial_{k}\hat{Q}_{1}(\zeta,s)ds\right\vert ~. \label{eq:eta2ProofB2}%
\end{align}
For the term (\ref{eq:eta2ProofB1}) we get, for $|k|\leq1$,%
\begin{align*}
&  \left\vert \frac{1}{2}e^{-\kappa t}\int_{1}^{t}(g_{1,1}(k,\sigma
)-2\frac{\kappa}{ik}\kappa\sigma+\frac{2|k|}{\sqrt{-ik}}\sigma)\hat{Q}%
_{1}(k,s)ds\right\vert \\
&  \leq\mathrm{const.}~e^{\Lambda_{-}t}\int_{1}^{\frac{t+1}{2}}|\Lambda
_{-}|^{2}(\sigma^{2}+s)\sigma e^{|\Lambda_{-}|\sigma}\mu_{1}^{II}(k,s)ds\\
&  +\mathrm{const.}~e^{\Lambda_{-}t}\int_{\frac{t+1}{2}}^{t}(1+|\Lambda
_{-}|)\sigma e^{|\Lambda_{-}|\sigma}\mu_{1}^{II}(k,s)ds\\
&  \leq\mathrm{const.}\left(  \frac{1}{t^{2-\delta}}\tilde{\mu}_{\alpha
^{\prime\prime}}(k,t)+\frac{1}{t^{5/2-\delta}}\bar{\mu}_{\alpha^{\prime\prime
}}(k,t)+\frac{1}{t^{2}}\tilde{\mu}_{\alpha^{\prime\prime}}(k,t)\right)  ~,
\end{align*}
and for $|k|>1$
\begin{align*}
&  \left\vert \frac{1}{2}e^{-\kappa t}\int_{1}^{t}(g_{1,1}(k,\sigma
)-2\frac{\kappa}{ik}\kappa\sigma+\frac{2|k|}{\sqrt{-ik}}\sigma)\hat{Q}%
_{1}(k,s)ds\right\vert \\
&  \leq\mathrm{const.}~e^{\Lambda_{-}t}\int_{1}^{\frac{t+1}{2}}|\Lambda
_{-}|^{3}\sigma\left(  \sigma^{2}+s\right)  e^{|\Lambda_{-}|\sigma}\mu
_{1}^{II}(k,s)ds\\
&  +\mathrm{const.}~e^{\Lambda_{-}t}\int_{\frac{t+1}{2}}^{t}(1+|\Lambda
_{-}|s)e^{|\Lambda_{-}|\sigma}\mu_{1}^{II}(k,s)ds\\
&  \leq\mathrm{const.}\left(  \frac{1}{t^{3-\delta}}\tilde{\mu}_{\alpha
^{\prime\prime}}(k,t)+\frac{1}{t^{7/2-\delta}}\bar{\mu}_{\alpha^{\prime\prime
}}(k,t)+\frac{1}{t^{3}}\tilde{\mu}_{\alpha^{\prime\prime}}(k,t)\right)  ~.
\end{align*}
The term (\ref{eq:eta2ProofB2}) is bounded using (\ref{eq:kqexpsquareikt}), so
that we get%
\[
\left\vert e^{-\kappa t}\int_{1}^{t}\frac{|k|k}{\sqrt{-ik}}\sigma\partial
_{k}\hat{Q}_{1}(\zeta,s)ds\right\vert \leq\mathrm{const.}\left\vert
e^{-\sqrt{-ik}t}\right\vert |k|^{3/2}\log(1+t)\in\mathcal{B}_{\alpha
^{\prime\prime},\infty,3-\delta}~,
\]
and thus $\hat{\eta}_{2}^{r,2}\in\mathcal{B}_{\alpha^{\prime\prime},\frac
{5}{2}-\delta,2-\delta}$.

To bound $\hat{\eta}_{3}^{r,2}$ we can rearrange the terms and use
(\ref{eq:meanValueTheoremQ1}) to get%
\begin{align*}
\hat{\eta}_{3}^{r,2}  &  =e^{-\kappa t}\left(  e^{\kappa}-1\right)  \int
_{1}^{t}\frac{\kappa}{ik}\kappa\sigma\hat{Q}_{1}(k,s)ds-e^{-\sqrt{-ik}t}%
\int_{1}^{t}\frac{ik}{\sqrt{-ik}}\sigma\hat{Q}_{1}(0,s)ds\\
&  =\left(  e^{-\kappa t}-e^{-\sqrt{-ik}t}\right)  \int_{1}^{t}\frac{ik}%
{\sqrt{-ik}}\sigma\hat{Q}_{1}(0,s)ds\\
&  +e^{-\kappa t}\int_{1}^{t}\left(  \frac{\kappa^{2}}{ik}\left(  e^{\kappa
}-1\right)  -\frac{ik}{\sqrt{-ik}}\right)  \sigma\hat{Q}_{1}(k,s)ds\\
&  +e^{-\kappa t}\int_{1}^{t}\frac{ik}{\sqrt{-ik}}k\sigma\partial_{k}\hat
{Q}_{1}(\zeta,s)ds~,
\end{align*}
for some $\zeta\in\lbrack0,k]$. We then get, using for the first term
(\ref{eq:exp(sqrt(-ik)-Kappa)}),%
\begin{align*}
|\hat{\eta}_{3}^{r,2}|  &  \leq\mathrm{const.}\left\vert e^{-\sqrt{-ik}%
t}\right\vert |k|^{3/2}t|k|^{1/2}\left\vert \int_{1}^{t}\sigma\hat{Q}%
_{1}(0,s)ds\right\vert \\
&  +\mathrm{const.}\left\vert e^{-\kappa t}\int_{1}^{t}\sqrt{-ik}%
k\sigma\partial_{k}\hat{Q}_{1}(\zeta,s)ds\right\vert \\
&  +\mathrm{const.}\left\vert e^{-\kappa t}\int_{1}^{t}\left(  \frac
{\kappa^{2}}{ik}\left(  e^{\kappa}-1\right)  +\sqrt{-ik}\right)  \sigma\hat
{Q}_{1}(k,s)ds\right\vert ~.
\end{align*}
For the third term we use%
\[
\left\vert \frac{\kappa^{2}}{ik}\left(  e^{\kappa}-1\right)  +\sqrt
{-ik}\right\vert =\left\vert -ik\left(  e^{\kappa}-1\right)  -\left(
e^{\kappa}-1-\kappa\right)  -\kappa+\sqrt{-ik}\right\vert
\]
which is bounded above, using (\ref{eq:kappa-sqrt(-ik)Bound}), by%
\[
\mathrm{const.}\left(  |k||\Lambda_{-}|e^{|\Lambda_{-}|}+|\Lambda_{-}%
|^{2}e^{|\Lambda_{-}|}+\min\{|\Lambda_{-}|^{2},|\Lambda_{-}|^{3}\}\right)
\leq\mathrm{const.}~|\Lambda_{-}|^{2}e^{|\Lambda_{-}|}~.
\]
We therefore have
\begin{align*}
|\hat{\eta}_{3}^{r,2}|  &  \leq\mathrm{const.}\left\vert e^{-\sqrt{-ik}%
t}\right\vert (|k|^{2}t+|k|^{3/2}\log(1+t))\\
&  +\mathrm{const.}~e^{\Lambda_{-}t}e^{|\Lambda_{-}|}\int_{1}^{t}|\Lambda
_{-}|^{2}\sigma e^{|\Lambda_{-}|\sigma}\mu_{1}^{II}(k,s)ds~,
\end{align*}
which, due to (\ref{eq:kqexpsquareikt}), Propositions~\ref{prop:sgL1} and
\ref{prop:sgL2}, and using (\ref{eq:kfortinmu}) to trade, where appropriate,
one factor of $|\Lambda_{-}|$ for a factor $t^{-1}$, shows that $\hat{\eta
}_{3}^{r,2}\in\mathcal{B}_{\alpha^{\prime\prime}-1,4-\delta,2}$.

To bound $\hat{\eta}_{4}^{r,2}$ we rearrange the terms using
(\ref{eq:meanValueTheoremQ1}), for some $\zeta\in\lbrack0,k]$, such that%
\begin{align*}
|\hat{\eta}_{4}^{r,2}|  &  =\left\vert e^{-\kappa t}\int_{1}^{t}\left(
\frac{\kappa}{ik}\kappa\hat{Q}_{1}(k,s)+\hat{Q}_{1}(0,s)\right)  \sigma
ds\right\vert \\
&  \leq\mathrm{const.}\left\vert e^{-\sqrt{-ik}t}\right\vert \left\vert
\int_{1}^{t}\left(  \frac{\kappa^{2}}{ik}+1\right)  \sigma\hat{Q}%
_{1}(0,s)ds\right\vert \\
&  +\mathrm{const.}\left\vert e^{-\sqrt{-ik}t}\right\vert \left\vert \int
_{1}^{t}\frac{\kappa^{2}}{ik}k\sigma\partial_{k}\hat{Q}_{1}(\zeta
,s)ds\right\vert \\
&  \leq\mathrm{const.}\left\vert e^{-\sqrt{-ik}t}\right\vert \left(
|k|+\left(  |k|+|k|^{2}\right)  \log(1+t)\right)  ~.
\end{align*}
We then use (\ref{eq:kqexpsquareikt}) to show that $\hat{\eta}_{4}^{r,2}%
\in\mathcal{B}_{\alpha^{\prime\prime},\infty,2-\delta}$.

For $\hat{\eta}_{5}^{r,2}$ we use (\ref{eq:exp(sqrt(-ik)-Kappa)}) and
(\ref{eq:kqexpsquareikt}), so that%
\[
|\hat{\eta}_{5}^{r,2}|=\left\vert \left(  e^{-\kappa t}-e^{-\sqrt{-ik}%
t}\right)  \int_{1}^{t}\left(  1+\frac{|k|}{\sqrt{-ik}}\right)  \sigma\hat
{Q}_{1}(0,s)ds\right\vert \leq\mathrm{const.}\left\vert e^{-\sqrt{-ik}%
t}\right\vert |k|^{3/2}t(1+|k|^{1/2})\in\mathcal{B}_{\alpha^{\prime\prime
},\infty,2}~.
\]

Finally, using (\ref{eq:kqexpsquareikt}) to bound $\hat{\eta}_{6}^{r,2}$, we
get%
\[
|\hat{\eta}_{6}^{r,2}|=\left\vert e^{-\sqrt{-ik}t}\int_{t}^{\infty}\left(
1+\frac{|k|-ik}{\sqrt{-ik}}\right)  \sigma\hat{Q}_{1}(0,s)ds\right\vert
\leq\mathrm{const.}\left\vert e^{-\sqrt{-ik}t}\right\vert (1+|k|^{1/2}%
)\frac{1}{t^{2}}\in\mathcal{B}_{\alpha^{\prime\prime},\infty,2}~.
\]

Gathering the bounds on the $\hat{\eta}_{i}^{r,2}$ terms yields
(\ref{eq:eta2RemainderSpace}), and by the opening remark of the proof also
(\ref{eq:omega2RemainderSpace}).
\end{proof}

\appendix

\section{Appendix}

\subsection{\label{sec:explicitasfun}Explicit expressions for the asymptotes}

The following are explicit functions for which Theorem~\ref{thm:mainresult} is
true:%
\begin{align}
\varphi_{1}(z) &  =-\frac{1}{4\sqrt{\pi}}\frac{r+1-z^{2}+zr+2z}{r^{3}%
\sqrt{r+1}}~,\label{eq:asphi1}\\
\psi_{1}(z) &  =-\frac{1}{4\sqrt{\pi}}\frac{r+1-z^{2}-zr-2z}{r^{3}\sqrt{r+1}%
}~,\label{eq:aspsi1}\\
\varphi_{2,1}(z) &  =-\frac{1}{\pi}\frac{2z}{r^{4}}~,\label{eq:asphi21}\\
\varphi_{2,2}(z) &  =\frac{1}{2\pi}\frac{1-z^{2}}{r^{4}}~,\label{eq:asphi22}\\
\psi_{2,1}(z) &  =-\frac{1}{\pi}\frac{1-z^{2}}{r^{4}}~,\label{eq:aspsi21}\\
\psi_{2,2}(z) &  =-\frac{1}{2\pi}\frac{2z}{r^{4}}~,\label{eq:aspsi22}\\
\eta_{W}(z) &  =-\frac{1}{2\sqrt{\pi z^{3}}}\left\{
\begin{array}
[c]{cc}%
e^{-1/4z}, & z\geq0\\
0, & z<0
\end{array}
\right.  ~,\label{eq:asetaW}\\
\omega_{W}(z) &  =\frac{1}{4\sqrt{\pi z^{5}}}%
\begin{cases}
(1-2z)e^{-1/4z}, & z\geq0\\
0, & z<0
\end{cases}
~,\label{eq:asomegaW}\\
\eta_{B}(z) &  =-\frac{1}{4\pi z^{3}}%
\begin{cases}
2z+\sqrt{\pi|z|}(1-2z)e^{-1/4z}(1-\operatorname{erfi}(1/\sqrt{4|z|})), &
z\geq0\\
2z+\sqrt{\pi|z|}(1-2z)e^{-1/4z}(1-\operatorname{erf}(1/\sqrt{4|z|})), & z<0
\end{cases}
~,\label{eq:asetaB}\\
\omega_{B}(z) &  =\frac{1}{8\pi z^{4}}%
\begin{cases}
2z(1-4z)+\sqrt{\pi|z|}(1-6z)e^{-1/4z}(1-\operatorname{erfi}(1/\sqrt{4|z|})), &
z\geq0\\
2z(1-4z)+\sqrt{\pi|z|}(1-6z)e^{-1/4z}(1-\operatorname{erf}(1/\sqrt{4|z|})), &
z<0
\end{cases}
~,\label{eq:asomegaB}%
\end{align}
where%
\[
r=\sqrt{1+z^{2}}~.
\]
These functions are obtained by taking the inverse Fourier transform of the
asymptotic terms calculated in Section~\ref{sec:asymptoticTerms}. 

\subsection{\label{sec:technical}Technical aspects of computations}

\subsubsection*{Mean-value theorem applied to $\hat{Q}_{1}$}

Applying the mean-value theorem in the variable $k$ we have%
\begin{equation}
\hat{Q}_{1}(k,s)=\hat{Q}_{1}(0,s)+k\partial_{k}\hat{Q}_{1}(\zeta,s)~,
\label{eq:meanValueTheoremQ1}%
\end{equation}
with some $\zeta\in\lbrack0,k]$ and (see
\cite{Boeckle.Wittwer-Decayestimatessolutions2011})%
\begin{equation}
\partial_{k}\hat{Q}_{1}\in\mathcal{B}_{\alpha,\frac{3}{2},2}~.
\label{eq:dkomegaBound}%
\end{equation}
The bound on $\partial_{k}\hat{Q}_{1}$ is improved in
Sections~\ref{sec:improvedkQ1} and \ref{sec:finalImprovements}, where it is
proved that this function is in $\mathcal{B}_{\alpha^{\prime\prime},2,2}$ and
$\mathcal{B}_{\alpha^{\prime\prime},\frac{5}{2}-\delta,2}$, respectively.

\subsubsection*{Inequalities for $k$ and $\kappa$}

Since $\kappa=\sqrt{k^{2}-ik}$ and $\Lambda_{-}=-\operatorname{Re}%
(\kappa)=-\frac{1}{2}\sqrt{2\sqrt{k^{2}+k^{4}}+2k^{2}}$, we have%
\[
|\kappa|=(k^{2}+k^{4})^{1/4}\leq|k|^{1/2}+|k|\leq2^{3/4}|\kappa|\leq
2^{3/4}(1+|k|)~,
\]
and that%
\[
|k|\leq|\Lambda_{-}|\leq|\kappa|\leq\sqrt{2}|\Lambda_{-}|~,
\]
from which we get, for $\sigma\geq0$,%
\[
e^{\Lambda_{-}\sigma}\leq e^{-|k|\sigma}~.
\]
The following inequalities are used throughout the proofs%
\begin{align}
\left\vert \kappa-\sqrt{-ik}\right\vert  &  =\left\vert \frac{k^{2}%
-ik-(-ik)}{\sqrt{k^{2}-ik}+\sqrt{-ik}}\right\vert \leq\frac{k^{2}}{2\left\vert
\sqrt{-ik}\right\vert }\nonumber\\
&  \leq\mathrm{const.}~|k|^{3/2}\leq\mathrm{const.}\min\{|\Lambda_{-}%
|^{2},|\Lambda_{-}|^{3}\}~, \label{eq:kappa-sqrt(-ik)Bound}%
\end{align}
and%
\begin{align}
\left\vert e^{-\kappa t}-e^{-\sqrt{-ik}t}\right\vert  &  \leq\left\vert
e^{-\sqrt{-ik}t}\left(  e^{(\sqrt{-ik}-\kappa)t}-1\right)  \right\vert
\nonumber\\
&  \leq\mathrm{const.}~\left\vert e^{-\sqrt{-ik}t}\right\vert \left\vert
\sqrt{-ik}-\kappa\right\vert t\nonumber\\
&  \leq\mathrm{const.}~\left\vert e^{-\sqrt{-ik}t}\right\vert |k|^{3/2}t~.
\label{eq:exp(sqrt(-ik)-Kappa)}%
\end{align}

\subsubsection*{Some inequalities for $\bar{\mu}_{\alpha}$ and $\tilde{\mu
}_{\alpha}$}

Using the notation introduced in Definition~\ref{def:mu}, we have for
$\alpha\geq0$ and $1\leq t<2$,%
\begin{align*}
\bar{\mu}_{\alpha}(k,t)  &  \leq\mathrm{const.}~\tilde{\mu}_{\alpha}%
(k,t)\leq\mathrm{const.}\\
\tilde{\mu}_{\alpha}(k,t)  &  \leq\mathrm{const.}~\bar{\mu}_{\alpha}%
(k,t)\leq\mathrm{const.}%
\end{align*}
and that for $t\geq2$ and $\beta\geq0$,%
\begin{align*}
e^{-|k|(t-1)}\mu_{\alpha,r}(k,t)  &  \leq\mathrm{const.}~e^{-|k|(t-1)}%
\leq\mathrm{const.}~\bar{\mu}_{\beta}(k,t)~,\\
e^{\Lambda_{-}(t-1)}\mu_{\alpha,r}(k,t)  &  \leq\mathrm{const.}~e^{\Lambda
_{-}(t-1)}\leq\mathrm{const.}~\tilde{\mu}_{\beta}(k,t)~,
\end{align*}
such that we have, for all $t\geq0$,%
\begin{align}
e^{-|k|(t-1)}\mu_{\alpha,r}(k,t)  &  \leq\mathrm{const.}~\bar{\mu}_{\alpha
}(k,t)~,\label{eq:mutomubar}\\
e^{\Lambda_{-}(t-1)}\mu_{\alpha,r}(k,t)  &  \leq\mathrm{const.}~\tilde{\mu
}_{\alpha}(k,t)~. \label{eq:mutomutilde}%
\end{align}
Another important inequality used in the proofs is that, for $p\geq0$,%
\begin{equation}
|k|^{p}\mu_{\alpha,r}\left(  k,t\right)  \leq\frac{\mathrm{const.}}{t^{rp}}%
\mu_{\alpha-p,r}\left(  k,t\right)  ~, \label{eq:kfortinmu}%
\end{equation}
which is due to the fact that%
\[
|k|^{p}\mu_{\alpha,r}\left(  k,t\right)  =\frac{t^{rp}}{t^{rp}}\frac{|k|^{p}%
}{1+(|k|t^{r})^{\alpha}}\leq\frac{\mathrm{const.}}{t^{rp}}\frac{1}%
{1+(|k|t^{r})^{\alpha-p}}~.
\]

\subsubsection*{Function spaces for some exponential functions}

\begin{proposition}
For $\alpha\geq1$, $p$, $q\geq0$, we have%
\begin{align}
k^{p}e^{-|k|t}  &  \in\mathcal{B}_{\alpha,p,\infty}~,p\geq
0~,\label{eq:kpexpabskt}\\
k^{q}e^{-\sqrt{-ik}t}~,~k^{q}e^{-\kappa t}  &  \in\mathcal{B}_{\alpha
,\infty,2q}~,q\geq0~. \label{eq:kqexpsquareikt}%
\end{align}

\end{proposition}

\begin{proof}
Using Definition~\ref{def:BapqSpaces} for functions belonging in
$\mathcal{B}_{\alpha,p,\infty}$ spaces, we must have%
\[
\sup_{t\geq1}\sup_{k\in\mathbb{R}\setminus\{0\}}\frac{|k^{p}e^{-|k|t}|}%
{\frac{1}{t^{p}}\bar{\mu}_{\alpha}(k,t)}=\sup_{t\geq1}\sup_{k\in
\mathbb{R}\setminus\{0\}}(|k|t)^{p}(1+(|k|t)^{\alpha})e^{-|k|t}<\infty~.
\]
We use the change of variable $z=kt$, so that%
\[
\sup_{t\geq1}\sup_{z\in\mathbb{R}\setminus\{0\}}|z|^{p}(1+|z|^{\alpha
})e^{-|z|}<\infty~.
\]

Similarly we have%
\[
\sup_{t\geq1}\sup_{k\in\mathbb{R}\setminus\{0\}}\frac{|k^{q}e^{-|k|t}|}%
{\frac{1}{t^{2q}}\tilde{\mu}_{\alpha}(k,t)}=\sup_{t\geq1}\sup_{k\in
\mathbb{R}\setminus\{0\}}(|k|t^{2})^{q}(1+(|k|t^{2})^{\alpha})\left\vert
e^{-\sqrt{-ikt^{2}}}\right\vert ~,
\]
and using the change of variable $z=kt^{2}$, we get%
\[
\sup_{t\geq1}\sup_{z\in\mathbb{R}\setminus\{0\}}|z|^{q}(1+|z|^{\alpha
})e^{-\sqrt{|z|/2}}<\infty~.
\]
For the functions $k^{q}e^{-\kappa t}$ we have%
\begin{align*}
\sup_{t\geq1}\sup_{k\in\mathbb{R}\setminus\{0\}}\frac{|k^{q}e^{-\kappa t}%
|}{\frac{1}{t^{2q}}\tilde{\mu}_{\alpha}(k,t)}  &  =\sup_{t\geq1}\sup
_{k\in\mathbb{R}\setminus\{0\}}(|k|t^{2})^{q}(1+(|k|t^{2})^{\alpha}%
)e^{\Lambda_{-}t}\\
&  \leq\sup_{t\geq1}\sup_{k\in\mathbb{R}\setminus\{0\}}(|\Lambda_{-}%
|t)^{2q}(1+((|\Lambda_{-}|t)^{2\alpha})e^{\Lambda_{-}t}~,
\end{align*}
and with the change of variable $z=|\Lambda_{-}|t$%
\[
\sup_{t\geq1}\sup_{z\in\mathbb{R}\setminus\{0\}}|z|^{2q}(1+|z|^{2\alpha
})e^{-z}<\infty~.
\]

\end{proof}

\subsection{Bounds on convolution\label{sec:newConv}}

We present variants of Proposition~9 and Corollary~10 from
\cite{Hillairet.Wittwer-Existenceofstationary2009}, which give bounds on
convolution products in $\mathcal{B}_{\alpha,p,q}$ spaces.

\begin{proposition}
[convolution]\label{prop:optConvBapq}Let $\alpha>1$, $s\geq r\geq0,$ and let
$a$, $b$ be continuous functions from $\mathbb{R}\setminus\{0\}\times
\lbrack1,\infty)$ to $\mathbb{C}$ satisfying the bounds,
\begin{align*}
\left\vert a(k,t)\right\vert  &  \leq\mu_{\alpha,r}(k,t)~,\\
\left\vert b(k,t)\right\vert  &  \leq\mu_{\alpha,s}(k,t)~,
\end{align*}
with $\mu_{a,r}$ and $\mu_{\alpha,s}$ as given in Definition~\ref{def:mu}.
Then, the convolution $a\ast b$ is a continuous function from $\mathbb{R}%
\times\lbrack1,\infty)$ to $\mathbb{C}$ and we have the bound%
\begin{equation}
\left\vert \left(  a\ast b\right)  (k,t)\right\vert \leq\mathrm{const.}%
~\frac{1}{t^{s}}\mu_{\alpha,r}\left(  k,t\right)  ~, \label{eq:optConvBapq}%
\end{equation}
uniformly in $t\geq1$, $k\in\mathbb{R}$.
\end{proposition}

\begin{proof}
We begin by splitting the integration interval into three sub-intervals, so
that%
\begin{align*}
2\pi\left\vert \left(  a\ast b\right)  (k,t)\right\vert  &  \leq\int_{-\infty
}^{\infty}\mu_{\alpha,r}\left(  k^{\prime},t\right)  \mu_{\alpha,s}\left(
k-k^{\prime},t\right)  dk^{\prime}=\\
&  =\int_{-\infty}^{-k/2}\ldots dk^{\prime}+\int_{k/2}^{\infty}\ldots
dk^{\prime}+\int_{-k/2}^{k/2}\ldots dk^{\prime}~,
\end{align*}
where we only consider $k>0$ since the functions $\mu_{\alpha,r}$ and
$\mu_{\alpha,s}$ are even with respect to $k$. We first note that%
\begin{align*}
&  \int_{-\infty}^{-k/2}\mu_{\alpha,r}\left(  k^{\prime},t\right)  \mu
_{\alpha,s}\left(  k-k^{\prime},t\right)  dk^{\prime}+\int_{k/2}^{\infty}%
\mu_{\alpha,r}\left(  k^{\prime},t\right)  \mu_{\alpha,s}\left(  k-k^{\prime
},t\right)  dk^{\prime}\\
&  \leq\mathrm{const.}~\mu_{\alpha,r}(\pm k/2,t)\int_{\mathbb{R}}\mu
_{\alpha,s}\left(  k-k^{\prime},t\right)  dk^{\prime}\leq\frac{\mathrm{const.}%
}{t^{s}}\mu_{\alpha,r}(k,t)~,
\end{align*}
where the factor $t^{-s}$ arises from the change of variables used in the
integral. For $kt^{r}\leq1$, we have $\frac{1}{2}\leq\mu_{\alpha,r}\leq1$, so
that%
\[
\int_{-k/2}^{k/2}\mu_{\alpha,r}\left(  k^{\prime},t\right)  \mu_{\alpha
,s}\left(  k-k^{\prime},t\right)  dk^{\prime}\leq\int_{\mathbb{R}}\mu
_{\alpha,s}\left(  k-k^{\prime},t\right)  dk^{\prime}\leq(\mathrm{const.}%
~\mu_{\alpha,r}(k,t))\cdot\frac{\mathrm{const.}}{t^{s}}~.
\]
For $kt^{r}>1$, we also have $kt^{s}>1$, and furthermore%
\[
\frac{\mu_{\alpha,s}(k,t)}{\mu_{\alpha,r}(k,t)}=\frac{1+(|k|t^{r})^{\alpha}%
}{1+(|k|t^{s})^{\alpha}}\leq\frac{2(|k|t^{r})^{\alpha}}{(|k|t^{s})^{\alpha}%
}=2t^{\alpha(r-s)}~,
\]
which shows that%
\[
\int_{-k/2}^{k/2}\mu_{\alpha,r}\left(  k^{\prime},t\right)  \mu_{\alpha
,s}\left(  k-k^{\prime},t\right)  dk^{\prime}\leq\mu_{\alpha,s}(k/2,t)\int
_{\mathbb{R}}\mu_{\alpha,r}\left(  k^{\prime},t\right)  dk^{\prime}\leq
\mu_{\alpha,r}(k/2,t)2t^{\alpha(r-s)}\frac{\mathrm{const.}}{t^{r}}~,
\]
which, since $\alpha>1$ and $s\geq r$, is bounded by a multiple of
$\mu_{\alpha,r}(k,t)/t^{s}$. Gathering the bounds yields (\ref{eq:optConvBapq}).
\end{proof}

\begin{corollary}
\label{corr:optConvBapq}Let $\alpha_{i}>1$, and, for $i=1,2$ let $p_{i}%
,q_{i}\geq0$. Let $\hat{f}_{i}\in\mathcal{B}_{\alpha_{i},p_{i},q_{i}}$, and
let%
\begin{align*}
\alpha &  =\min\{\alpha_{1},\alpha_{2}\}~,\\
p  &  =\min\{p_{1}+p_{2}+1,p_{1}+q_{2}+2,p_{2}+q_{1}+2\}~,\\
q  &  =q_{1}+q_{2}+2~.
\end{align*}
Then $\hat{f}_{1}\ast\hat{f}_{2}\in\mathcal{B}_{\alpha,p,q}$ and there exists
a constant $C$, dependent only on $\alpha_{i}$, such that%
\[
\left\Vert \hat{f}_{1}\ast\hat{f}_{2};\mathcal{B}_{\alpha,p,q}\right\Vert \leq
C\left\Vert \hat{f}_{1};\mathcal{B}_{\alpha_{1},p_{1},q_{1}}\right\Vert
\cdot\left\Vert \hat{f}_{2};\mathcal{B}_{\alpha_{2},p_{2},q_{2}}\right\Vert
~.
\]

\end{corollary}

\begin{proof}
Using that $\mathcal{B}_{\alpha_{i},p_{i},q_{i}}\subset\mathcal{B}%
_{\min\{\alpha_{1},\alpha_{2}\},p_{i},q_{i}}$, this is an immediate
consequence of Proposition~\ref{prop:optConvBapq}.
\end{proof}

\begin{proposition}
[convolution with $|\kappa|^{-1}$ discontinuity]\label{prop:convdisc}Let
$\alpha_{i}>1$, and, for $i=1,2$ let $p_{i},q_{i}\geq0$. Let $\hat{f}%
\in\mathcal{B}_{\alpha_{1},p_{1},q_{1}}$ and $\kappa\cdot\hat{g}\in
\mathcal{B}_{\alpha_{2},p_{2},q_{2}}$, and let%
\begin{align*}
\alpha &  =\min\{\alpha_{1,}\alpha_{2}\}~,\\
p  &  =\min\{p_{1}+p_{2}+\frac{1}{2},p_{1}+q_{2}+1\}~,\\
q  &  =\min\{q_{1}+p_{2}+\frac{1}{2},q_{1}+q_{2}+1\}~.
\end{align*}
Then $\hat{f}\ast\hat{g}\in\mathcal{B}_{\alpha,p,q}$ and there exists a
constant $C$, dependent only on $\alpha_{i}$, such that%
\[
\left\Vert \hat{f}\ast\hat{g};\mathcal{B}_{\alpha,p,q}\right\Vert \leq
C\left\Vert \hat{f};\mathcal{B}_{\alpha_{1},p_{1},q_{1}}\right\Vert
\cdot\left\Vert \hat{g};\mathcal{B}_{\alpha_{2},p_{2},q_{2}}\right\Vert ~.
\]

\end{proposition}

\begin{proof}
This proposition is a consequence of Proposition~11 of
\cite{Boeckle.Wittwer-Decayestimatessolutions2011}.
\end{proof}

\subsection{Convolution with the semi-groups $e^{\Lambda_{-}t}$ and
$e^{-|k|t}$}

In an effort of self-consistency, we present the results for the convolution
with the semi-groups $e^{\Lambda_{-}t}$ and $e^{-|k|t}$ which are all proved
in \cite{Hillairet.Wittwer-Existenceofstationary2009}. In order to bound the
integrals over the interval $[1,t]$ we systematically split them into
integrals over $[1,\frac{1+t}{2}]$ and integrals over $[\frac{1+t}{2},t]$ and
bound the resulting terms separately. The range for the parameter $\beta$ has
been extended to include values between $0$ and $1$ using H\"{o}lder's
inequality in the propositions for the intervals $[(t+1)/2,t]$ and
$[t,\infty)$. In practice, when a logarithmic bound is found we use that for
all $\delta\in(0,1)$ there exists a constant such that%
\begin{equation}
\log\left(  1+t\right)  \leq\mathrm{const.}~t^{\delta}~, \label{eq:boundlog}%
\end{equation}
in order to present a bound in terms of $\mathcal{B}_{\alpha,p,q}$ spaces.

For the semi-group $e^{\Lambda_{-}t}$ we have:

\begin{proposition}
\label{prop:sgL1}Let $\alpha\geq0$, $r\geq0$ and $\delta\geq0$ and
$\gamma+1\geq\beta\geq0$. Then,%
\begin{align*}
&  e^{\Lambda_{-}(t-1)}\int_{1}^{\frac{t+1}{2}}e^{|\Lambda_{-}|(s-1)}%
|\Lambda_{-}|^{\beta}\frac{(s-1)^{\gamma}}{s^{\delta}}\mu_{\alpha,r}(k,s)ds\\
&  \leq\left\{
\begin{array}
[c]{l}%
\displaystyle\mathrm{const.}\frac{1}{t^{\beta}}\tilde{\mu}_{\alpha
}(k,t),\text{ if }\delta>\gamma+1\\
\\
\displaystyle\mathrm{const.}\frac{\log(1+t)}{t^{\beta}}\tilde{\mu}_{\alpha
}(k,t),\text{ if }\delta=\gamma+1\\
\\
\displaystyle\mathrm{const.}\frac{t^{\gamma+1-\delta}}{t^{\beta}}\tilde{\mu
}_{\alpha}(k,t),\text{ if }\delta<\gamma+1
\end{array}
\right.
\end{align*}
uniformly in $t\geq1$ and $k\in\mathbb{R}$.
\end{proposition}

\begin{proposition}
\label{prop:sgL2}Let $\alpha\geq0$, $r\geq0$, $\delta\in\mathbb{R}$, and
$\beta\in\lbrack0,1]$. Then,%
\[
e^{\Lambda_{-}(t-1)}\int_{\frac{t+1}{2}}^{t}e^{|\Lambda_{-}|(s-1)}|\Lambda
_{-}|^{\beta}\frac{1}{s^{\delta}}\mu_{\alpha,r}(k,s)ds\leq\frac
{\mathrm{const.}}{t^{\delta-1+\beta}}\mu_{\alpha,r}(k,t)~,
\]
uniformly in $t\geq1$ and $k\in\mathbb{R}$.
\end{proposition}

\begin{proposition}
\label{prop:sgL3}Let $\alpha\geq0$, $r\geq0$, $\delta>1$, and $\beta\in
\lbrack0,1]$. Then,%
\begin{align*}
e^{|\Lambda_{-}|(t-1)}\int_{t}^{\infty}e^{\Lambda_{-}(s-1)}|\Lambda
_{-}|^{\beta}\frac{1}{s^{\delta}}\mu_{\alpha,r}(k,s)ds  &  \leq\frac
{\mathrm{const.}}{t^{\delta-1+\beta}}\mu_{\alpha,r}(k,t)~,\\
\left\vert \frac{\kappa}{ik}\left(  e^{|\Lambda_{-}|\left(  t-1\right)
}-e^{\Lambda_{-}\left(  t-1\right)  }\right)  \right\vert \int_{t}^{\infty
}e^{\Lambda_{-}\left(  s-1\right)  }|\Lambda_{-}|^{\beta}\frac{1}{s^{\delta}%
}\mu_{\alpha,r}\left(  k,s\right)  ds  &  \leq\frac{\mathrm{const.}}%
{t^{\delta-2+\beta}}\mu_{\alpha,r}\left(  k,t\right)  ~,
\end{align*}
uniformly in $t\geq1$ and $k\in\mathbb{R}$.\bigskip
\end{proposition}

The results for the semi-group $e^{-|k|t}$ are very similar.

\begin{proposition}
\label{prop:sgk1}Let $\alpha\geq0$, $r\geq0$ and $\delta\geq0$ and
$\gamma+1\geq\beta\geq0$. Then,%
\begin{align*}
&  e^{-|k|(t-1)}\int_{1}^{\frac{t+1}{2}}e^{|k|(s-1)}|k|^{\beta}\frac
{(s-1)^{\gamma}}{s^{\delta}}\mu_{\alpha,r}(k,s)~ds\\
&  \leq\left\{
\begin{array}
[c]{l}%
\displaystyle\mathrm{const.}\frac{1}{t^{\beta}}\bar{\mu}_{\alpha}(k,t),\text{
if }\delta>\gamma+1\\
\\
\displaystyle\mathrm{const.}\frac{\log(1+t)}{t^{\beta}}\bar{\mu}_{\alpha
}(k,t),\text{ if }\delta=\gamma+1\\
\\
\displaystyle\mathrm{const.}\frac{t^{\gamma+1-\delta}}{t^{\beta}}\bar{\mu
}_{\alpha}(k,t),\text{ if }\delta<\gamma+1
\end{array}
\right.
\end{align*}
uniformly in $t\geq1$ and $k\in\mathbb{R}$.
\end{proposition}

\begin{proposition}
\label{prop:sgk2}Let $\alpha\geq0$, $r\geq0$, $\delta\in\mathbb{R}$, and
$\beta\in\lbrack0,1]$. Then,%
\[
e^{-|k|(t-1)}\int_{\frac{t+1}{2}}^{t}e^{|k|(s-1)}|k|^{\beta}\frac{1}%
{s^{\delta}}\mu_{\alpha,r}(k,s)~ds\leq\frac{\mathrm{const.}}{t^{\delta
-1+\beta}}\mu_{\alpha,r}(k,t)~,
\]
uniformly in $t\geq1$ and $k\in\mathbb{R}$.
\end{proposition}

\begin{proposition}
\label{prop:sgk3}Let $\alpha\geq0$, $r\geq0$, $\delta>1$, $\beta\in
\lbrack0,1]$. Then,%
\begin{align*}
e^{|k|(t-1)}\int_{t}^{\infty}e^{-|k|(s-1)}|k|^{\beta}\frac{1}{s^{\delta}}%
\mu_{\alpha,r}(k,s)~ds  &  \leq\frac{\mathrm{const.}}{t^{\delta-1+\beta}}%
\mu_{\alpha,r}(k,t)~,\\
\left\vert \frac{|k|}{ik}\left(  e^{|k|(t-1)}-e^{-|k|(t-1)}\right)
\right\vert \int_{t}^{\infty}e^{-|k|(s-1)}|k|^{\beta}\frac{1}{s^{\delta}}%
\mu_{\alpha,r}(k,s)~ds  &  \leq\frac{\mathrm{const.}}{t^{\delta-1+\beta}}%
\mu_{\alpha,r}(k,t)~,
\end{align*}
uniformly in $t\geq1$ and $k\in\mathbb{R}$.
\end{proposition}

\bibliographystyle{amsplain}
\bibliography{CompleteDatabase2011}

\end{document}